\documentclass[journal]{IEEEtran}
\usepackage{comment}
\usepackage{soul}
\soulregister{\cite}7 
\soulregister{\ref}7 
  
\usepackage[T1]{fontenc} 
\usepackage{color}
\usepackage{multirow}
\usepackage{array}
\usepackage{amsmath,amsthm,amsfonts,amssymb,bm,mathrsfs} 
\usepackage{url} 
\usepackage{soul} 
\usepackage{graphicx} 
\usepackage{stfloats}
\usepackage{enumitem}
\usepackage{algorithmic}
\usepackage{algorithm}
\usepackage{overpic}
\usepackage{cite} 
\usepackage{wasysym}
\usepackage{placeins}
\usepackage[colorlinks,citecolor=black, filecolor=black, linkcolor=black,urlcolor=black]{hyperref}
\usepackage[table]{xcolor}
\definecolor{slg}{gray}{0.95}

\allowdisplaybreaks[4]

\makeatletter
\newcommand{\removelatexerror}{\let\@latex@error\@gobble}
\makeatother
\usepackage{booktabs, makecell, tabularx, threeparttable}
\usepackage{caption}
\newcolumntype{C}{>{\centering\arraybackslash}X} 

\usepackage{stfloats}
\usepackage{siunitx}
\usepackage{multirow}
\usepackage{subfigure}

\newtheorem{assumption}{Assumption} 
\newtheorem{theorem}{Theorem}
\newtheorem{lemma}{Lemma}

\def\ie{{i.e.}}

\def\etal{\textit{et al.}~}
\def\FTX{\mathcal{F}_{\rm{TX}}}
\def\FRXN{{\mathcal{F}_{{\rm RX}{_n}}}}
\def\rxn{${\rm RX}_n$}
\def\tx{$\rm{TX}$}
\def\rxs{$\rm{RXs}$}
\def\rx{$\rm{RX}$}

\usepackage{hyperref} 

\begin{document}
\title{Reinforcement Learning-Based Heterogeneous Multi-Task Optimization in Semantic Broadcast Communications} 
\author{Zhilin~Lu,~Rongpeng~Li,~Zhifeng~Zhao,~and~Honggang~Zhang\vspace{-1cm}
    \thanks{
  Z.~Lu, and R. Li are with Zhejiang University, China (email: \{lu\_zhilin, lirongpeng\}@zju.edu.cn).} \thanks{Zhifeng Zhao is with Zhejiang Lab as well as Zhejiang University, China (email: zhaozf@zhejianglab.com).} \thanks{Honggang Zhang is with City University of Macau, China (email: hgzhang@cityu.edu.mo).}   
}
	\maketitle

\begin{abstract}
   Semantic broadcast communications (Semantic BC) for image transmission have achieved significant performance gains for single-task scenarios. Nevertheless, extending these methods to multi-task scenarios remains challenging, as different tasks typically require distinct objective functions, leading to potential conflicts within the shared encoder. In this paper, 
   we propose a tri-level reinforcement learning (RL)-based multi-task Semantic BC framework, termed \texttt{SemanticBC-TriRL}, which effectively resolves such conflicts and enables the simultaneous support of multiple downstream tasks at the receiver side, including image classification and content reconstruction tasks. Specifically, the proposed framework employs a bottom-up tri-level alternating learning strategy, formulated as a constrained multi-objective optimization problem. At the first level, task-specific decoders are locally optimized using supervised learning. At the second level, the shared encoder is updated via proximal policy optimization (PPO), guided by task-oriented rewards. At the third level, a multi-gradient aggregation-based task weighting module adaptively adjusts task priorities and steers the encoder optimization. Through this hierarchical learning process, the encoder and decoders are alternately trained, and the three levels are cohesively integrated via constrained learning objective. Besides, the convergence of \texttt{SemanticBC-TriRL} is also theoretically established. Extensive simulation results demonstrate the superior performance of the proposed framework across diverse channel conditions, particularly in low SNR regimes, and confirm its scalability with increasing numbers of receivers.
\end{abstract}

	
\IEEEpeerreviewmaketitle

\section{Introduction}
\label{sec: Introduction}
  \IEEEPARstart{A}{longside} the shift from connecting people to connecting intelligent unmanned systems \cite{liu2022task}, semantic communication (SemCom) \cite{lu2022semantics,lu_rethinking_2023,zhou_adaptive_2022,ni_topology_2025}, which primarily utilizes artificial intelligence (AI), especially deep neural networks (DNNs) as codecs to capture and deliver essential information, has attracted widespread attention from both the academic and industrial communities.  
  According to the task types at the receivers (\rxs), the existing works on SemCom can be broadly 
  categorized into semantic-level data reconstruction and effectiveness-level task execution  \cite{lu2022semantics}. The former case exemplified by joint source-channel coding (JSCC) approaches \cite{xie2021deep, xu2021wireless}, aims to recover the original semantic content as accurately as possible, while for the latter, task-relevant semantic features are extracted to support specific downstream tasks \cite{xie2022task}. Furthermore, multi-task semantic communication frameworks typically fall into several categories, including multiple tasks with separate codecs \cite{xie2022task,patwa2020semantic}, multiple tasks with a single generalized codec \cite{liu2022task,sheng2022multi}, and task-specific decoders paired with a shared encoder \cite{ma2023features,lu2024self,hu2022one}. Considering the superior efficiency at the computational resource-limited transmitter (\tx), such as autonomous vehicles and smart surveillance \cite{yu2025multi}, for delivering the same content to multiple \rxs~\cite{lu2024self}, we primarily focus on the third category, specifically termed as semantic broadcast communications (BC). 
  
  For semantic BC schemes, there have emerged several related studies \cite{hu2022one,sheng2022multi, lyu2024semantic, lu2024self}. For instance, Hu \etal \cite{hu2022one} design a one-to-many semantic BC system by utilizing task-oriented semantic recognizers DistilBERT to distinguish different \rxs~according to emotional properties (\ie, positive or negative). However, since the end-to-end joint training of the encoder and multiple decoders under the JSCC paradigm, the system exhibits limited scalability at the \rxs~sides. To overcome this limitation, our previous work \cite{lu2024self} introduces a reinforcement learning (RL)-based alternating optimization strategy, which enables the encoder and decoders to be updated independently via self-critical learning, thereby enhancing system scalability, compatibility, and convergence stability.

  Nevertheless, these works focus on executing homogeneous tasks at different \rxs~\cite{lu2024self,hu2022one,nguyen2024swin}, or collaboratively completing tasks by utilizing correlated information from multiple transmitters \cite{xu2023semantic,lo2023collaborative}. As a more general case, heterogeneous tasks such as object detection, image classification, and target tracking can occur concurrently. 
  Therefore, multiple \rxs~in semantic BC may inevitably accomplish different tasks based on the same information acquired from the \tx. Correspondingly, a critical challenge naturally arises. 
  Since different tasks correspond to different objective functions, a satisfactory encoding strategy shall competently tackle the potential conflict among multiple objective functions. Due to its blunt ignorance of complications, a simple, direct linear combination of multiple functions \cite{zhang2022unified, wu2024multi} may fail to find an optimal encoding strategy. Instead, beyond the training of DNN parameters, calibrating the relative importance shall be taken into account as well. Theoretically, such a multi-objective optimization falls into the scope of locating the Pareto optimality \cite{sener2018multi}, as conflicting gradients arising from divergent update directions for different tasks may significantly undermine the stability and convergence of the learning process at \tx. Therefore, there emerges a strong incentive to develop a holistic solution to incorporate codec training and adaptive task prioritization strategies to achieve optimal task accomplishment performance. 
  
  In this regard, the gradient-based methods that leverage available gradient information from DNN training have shown superiority in various multi-objective bi-level optimization works \cite{momma2022multi, ye2024first, sener2018multi}, promising a well-suited complement to our setting. To be more specific, within a semantic BC framework consisting of a shared encoder and task-specific decoders, the encoder optimization can naturally exploit the relationships among the heterogeneous decoder objectives. This enables adaptive adjustment of task-wise weight assignments, thereby mitigating potential gradient conflicts in the encoder updates arising from multiple objective functions at \rxs.

    \begin{table*}[t!]
		\centering
		\caption{Summary and comparison of multi-task SemCom schemes}    
		\label{tab: comparison-BC}	
		\begin{tabular}{m{1cm} m{1.5cm} m{2cm} m{2cm} m{1.5cm} m{7.5cm}} 		
			\hline
			Refs& Task heterogeneity&Decoupled codec training& Balanced task performance &Communication efficiency&Brief Description\\ 
			\hline
            \cite{sheng2022multi}& \CIRCLE & \Circle & $\RIGHTcircle$ & \CIRCLE & It lacks investigation into scalability with respect to the number of \rxs.\\
            \cite{ma2023features}& \CIRCLE & \Circle & $\RIGHTcircle$ & \CIRCLE& It involves a complex model architecture with a large number of parameters.\\
            \cite{lyu2024semantic} & $\CIRCLE$&\Circle &$\RIGHTcircle$ &$\RIGHTcircle$&The intrinsic connections among tasks are not fully considered. \\
            \cite{zhang2024unified} & $\CIRCLE$& \Circle & $\RIGHTcircle$&$\RIGHTcircle$& The unified codebook for feature representation needs to be redesigned when the task changes. \\
            \cite{sagduyu2023multi}&  $\RIGHTcircle$  & \Circle & $\RIGHTcircle$&$\RIGHTcircle$  & The task weights are assigned directly without further optimization.\\
            \cite{razlighi2024semantic}&$\RIGHTcircle$ &\Circle &  $\RIGHTcircle$ & $\RIGHTcircle$ & The intrinsic connections among tasks are not fully considered.\\
            \hline
            This work & \CIRCLE & \CIRCLE &  \CIRCLE &  \CIRCLE & By integrating intrinsic task relationships and relative importance, the encoder is optimized based on PPO approach, guided by a weight assignment module.\\
            \hline
            \multicolumn{6}{>{\footnotesize\itshape}r}{Notations: \rm{${\CIRCLE}$} \emph{indicates fully included;} \rm{${\RIGHTcircle}$} \emph{means partially included; $\Circle$ indicates not included.}}
		\end{tabular}
\end{table*}

  In this paper, we propose a tri-level RL-based multi-task semantic BC framework, \texttt{SemanticBC-TriRL}, which enables alternating updates of the encoder and decoders while dynamically balancing the weights of multiple objective functions at the \rxs. 
  To support flexible and independent optimization of the encoder and decoders, we build upon our previous work \cite{lu2024self}, and further introduce PPO to train the shared encoder in a task-oriented manner, guided by receiver-side performance requirements, thereby decoupling the training of the encoder and decoders.
  Furthermore, to address potential gradient conflicts among multiple task objectives during the encoder parameter updates, we incorporate a multi-gradient aggregation method that adaptively adjusts the loss weights assigned to each task. In this formulation, the task-specific decoders, the shared encoder, and the weight assignment module are jointly optimized by a tri-level alternating learning mechanism, corresponding to the first, second, and third levels, respectively. The optimization is carried out in a bottom-up manner, where the first level solutions implicitly guide the second level updates. By seamlessly integrating the optimization of encoder-decoders training with the task weight adjustment, \texttt{SemanticBC-TriRL} addresses the constrained tri-level learning problem effectively. A detailed comparison with closely related works is provided in Table~\ref{tab: comparison-BC}, and the major contributions of our work are three-fold:

  \begin{itemize}
  \item We put forward \texttt{SemanticBC-TriRL}, a tri-level RL-based one-to-many semantic BC framework for heterogeneous task execution. The learning process is formulated as a multi-objective optimization problem with corresponding constraints, in which the encoder, task-specific decoders, and weight assignment module are optimized by the alternate training in a bottom-up manner. 
  In this way, \texttt{SemanticBC-TriRL} seamlessly accommodates diverse objective functions and performance requirements across different \rxs~and allows for ready expansion to additional \rxs~with flexible scalability.
  \item At the first level, the decoders are updated locally based on the singleton condition \cite{ye2021multi,liu2020generic}. For the second level of encoder optimization, we resort to the PPO algorithm for its enhanced decision-making capabilities with the given optimized decoders. 
  \item At the third level, to mitigate potential conflicts across multiple tasks,
  the weight assignment module utilizes a multi-gradient aggregation method to adaptively adjust task weights. Besides, we provide the non-asymptotic convergence analysis for \texttt{SemanticBC-TriRL}.  
  \item Based on extensive simulation results, the proposed framework demonstrates significant performance improvement in the multi-task semantic BC system for both reconstruction and classification tasks. Moreover, \texttt{SemanticBC-TriRL} exhibits strong adaptability to varying channel conditions and consistently outperforms conventional methods such as ``BPG+LDPC'' and DNNs-based semantic models including ``JSCC'' \cite{bourtsoulatze2019deep} and ``Deep JSCC'' \cite{lyu2024semantic}. In addition, the results further attest to its scalability with an increased number of \rxs.

  \end{itemize} 
  
  The rest of this paper is mainly organized as follows. Section \ref{sec: Related Works} reviews the recent related works. In Section \ref{sec: System Model and Problem Formulation}, we present the framework of \texttt{SemanticBC-TriRL} and formulate the optimization problem. In Section \ref{sec: Architecture of the Proposed SemanticBC-TriRL}, we describe the details of the tri-level \texttt{SemanticBC-TriRL} and analyze the theoretical convergence. In Section \ref{sec: Simulation Results}, the numerical results are illustrated and discussed. Finally, conclusions are drawn in Section \ref{sec: Conclusions}.

\section{Related Works}
\label{sec: Related Works}

Most SemCom works are task-specific or task-oriented, ensuring that only the most crucial, pertinent, and beneficial information is effectively transmitted and applied in downstream tasks. Taking the example of task execution applications, SemCom learns to transmit task-specific semantics for decision-making at the \rxs \cite{xie2022task, jankowski2020wireless,hu2023robust,weng2023deep }. For instance, for image retrieval tasks under power and bandwidth constraints, \cite{jankowski2020wireless} proposes a robust JSCC model, which integrates a retrieval baseline with a feature encoder, followed by scalar quantization and entropy coding. In \cite{hu2023robust}, the authors design a novel masked vector quantized-variational autoencoder (VQ-VAE) with the noise-related masking strategy for image classification,  retrieval and reconstruction tasks. Except for image-oriented tasks, \cite{weng2023deep} develops a deep learning-based
SemCom system for speech recognition and synthesis tasks to achieve diverse system outputs. The abovementioned works mainly focus on enhancing the performance in single-user/task SemCom systems. But, communication challenges (e.g., scalability) emerge as the number of users increases.

In terms of multi-user/task SemCom systems \cite{zhang2024unified, kang2022task, cao2025task}, several recent efforts have sought to improve scalability and task adaptability. 
Zhang \etal \cite{zhang2024unified}  propose a unified multi-task SemCom solution that can adjust the number of transmitted symbols for six widely recognized tasks with different modalities. Furthermore, they design a unified codebook for feature representation, where only the indices of these task-specific features in the codebook are transmitted. \cite{kang2022task} develops a lightweight model on the frontend unmanned aerial vehicles (UAVs) for semantic block transmission with perception of images and channel conditions, and deep reinforcement learning (DRL) is employed to explore the semantic blocks that most contribute to the back-end classifier under various channel conditions. These approaches are predominantly task-specific and focus primarily on optimizing task performance at \rxs. However, these works ignore different requirements of diverse tasks deployed on multi-user systems \cite{wang2024feature} and tend to inadequately exploit the feedback information from \rxs~that are crucial for refining the encoder's performance. By contrast, \cite{wang2024feature} designs an importance-aware joint source-channel coding (I-JSCC) framework for task-oriented semantic communications, in which a joint semantic-channel transmission mechanism is designed by selectively transmitting task-important features to reduce communication overhead. \cite{li2023non} introduces a non-orthogonal multiple access (NOMA)-based multi-user SemCom system that supports diverse data modalities and enhances spectral efficiency by superimposing semantic symbols of multiple users. The system employs an asymmetric quantizer and a neural network for intelligent multi-user detection, demonstrating superior performance and robustness in low-to-medium SNR scenarios compared to traditional methods. 

Moreover, some multi-task frameworks \cite{sheng2022multi, sagduyu2023multi, razlighi2024semantic, ma2023features, gao2025multi, lyu2024semantic} leverage the shared encoder to reduce the number of transmissions required for task completion at different \rxs. Among these, some \cite{sheng2022multi, lyu2024semantic, gao2025multi} focus on performing multiple tasks at a single \rx, while others \cite{ma2023features, sagduyu2023multi, razlighi2024semantic} aim to broadcast the extracted semantic features to multiple \rxs, each executing a distinct task. 
In the former category, Shen \etal \cite{sheng2022multi} propose a multi-task SemCom system for textual tasks based on bidirectional encoder representations from transformers (BERT), which jointly performs text classification and regression tasks in an end-to-end training framework. However, their work lacks investigation into scalability with respect to the number of \rxs, which limits its applicability in more general multi-receiver scenarios. \cite{lyu2024semantic} designs a multi-task SemCom for image recovery and classification at the same time on top of a new training framework based on gated deep JSCC via domain randomization. Although their design demonstrates efficient support for multiple tasks, the intrinsic connections among tasks are not fully considered. For the latter category, \cite{sagduyu2023multi} designs a multi-task deep learning approach that involves training a common encoder at \tx~and individual decoders at the \rxs~for better accomplishing multiple image classification tasks, in which the overall loss is obtained by taking the weighted sum of individual tasks. However, they appear to have noticed the issue of tasks weight assignment during training, but did not conduct further investigation or optimization. 
Ma \etal \cite{ma2023features} propose a practical robust features-disentangled multi-user semantic BC framework, where \tx~employs a feature selection module to broadcast users' intended semantic features, while \rxs~utilize feature completion modules to capture the desired information, significantly improving transmission efficiency. However, both the feature selection and the feature completion rely on complex model architectures and a large number of parameters. 

Based on the above observations, most of the existing semantic BC frameworks do not take into account intrinsic connections and different importance of tasks at \rxs, and few works associated optimization mechanism to achieve consistent adaptation between the \tx~and \rxs. Addressing this challenge, our prior work \cite{lu2024self} regards the performance of the tasks at \rxs~as rewards, employing a lightweight self-critical RL approach to alternately optimize the encoder and decoders, with tasks uniformly weighted based on their importance level. Nevertheless, given that the semantic BC scenarios typically involve heterogeneous tasks, the update of the encoder during the training phase should integrate diverse requirements of these tasks, as potentially conflicting gradients arising from divergent update directions for different tasks may significantly undermine the stability and convergence of the learning process. Therefore, in heterogeneous multi-task semantic BC scenarios, the model design should holistically incorporate codec training and adaptive task prioritization strategies to achieve optimal task accomplishment performance.

\section{System Model and Problem Formulation} 
\label{sec: System Model and Problem Formulation}

\begin{figure*}[hbt]
   \centering
   \setlength{\abovecaptionskip}{0.2cm}
		\includegraphics[width=0.98\linewidth]{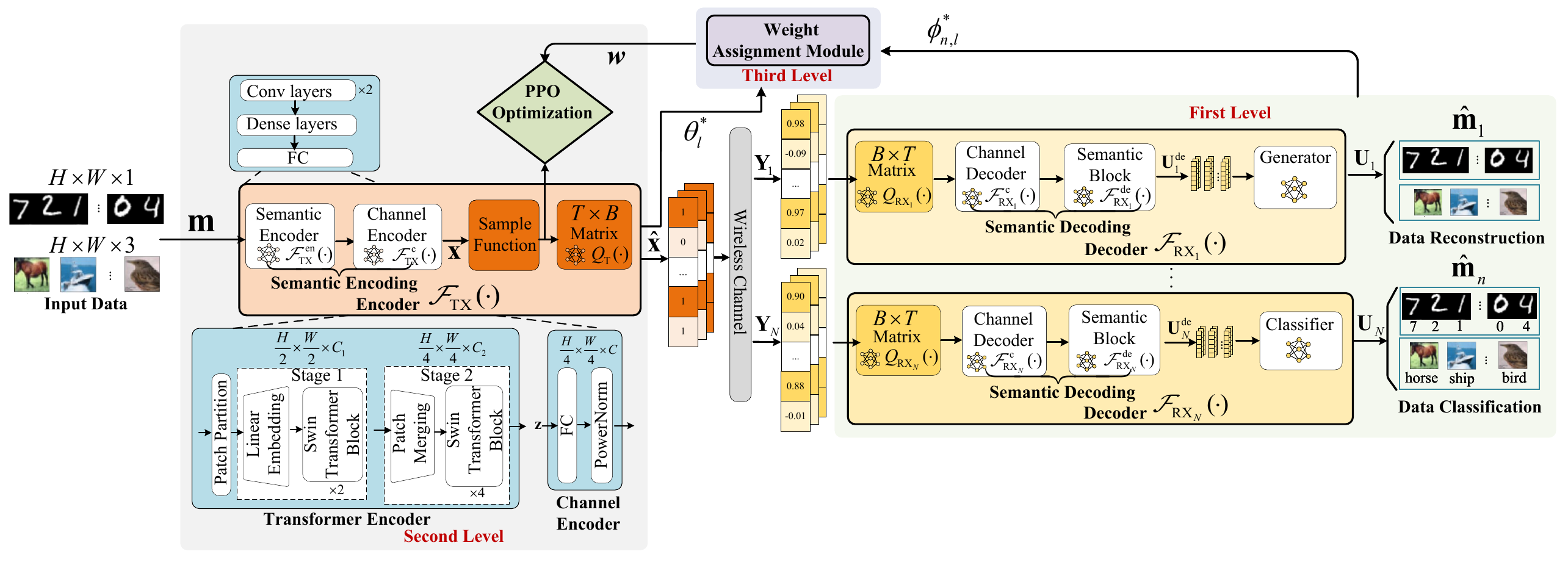}\\ 
\caption{The framework of the proposed multi-task semantic BC system.}
		\label{system_model}
\end{figure*}

\begin{table}[hbt]
    \centering
    \caption{Notions used in this paper.}
    \label{notation}
    \begin{tabular}{m{1.9cm} m{5.9cm}}
    \toprule   
    Notation      & Definition       \\
    \midrule
    $N$ & The number of \rxs\\
    $\FTX(\cdot), \FRXN(\cdot)$ &  The abbreviated terms for \tx~and \rxn~($n\sim N$) respectively\\ 
    $Q_{\rm{TX}}(\cdot), Q_{\rm{RX_n}}(\cdot)$ &The quantization and dequantization layers\\
    ${\boldsymbol{m}}, {{\hat{\boldsymbol{m}}}_n}$  &   Input message of \tx~ and output results of \rxn  \\
    $\Theta_n (\cdot)$ &  Task performance metric of \rxn\\ 
    ${\theta}, \phi_n$ & Parameters for semantic encoder and semantic decoder $n$\\
    $\boldsymbol w=\{ w_{n}\}_{n=1}^{N}$  & Learnable weight vector of multiple \rxs\\
    $v$& A variable denoted as $v\equiv (\boldsymbol{w},\theta)$ \\  
    $F(v)$& Third level function\\ 
    $g(v)$& Constraint function\\ 
    $\tilde g(v)$& Approximated constraint function\\ 
    $\kappa$ & Local iterations at each decoder\\
    $H$ & Inner steps at the encoder\\
    $L$ & Number of iterations in each epoch\\
    $T$&Batch size\\ 
    $\tilde g(v)$& Approximated constraint function\\
    $M$ & The upper bound of $\nabla F(v)$\\  
    $B$ & The fixed number of bits for a transmitted image\\
    \bottomrule  
    \end{tabular}  
\end{table}

\subsection{System Model}  

  Beforehand, Table \ref{notation} summarizes the key notations in this paper.
  As shown in Fig. \ref{system_model},  we consider a multi-task semantic BC scenario, involving a shared \tx~and $N$ task-specific \rxs~executing either classification or reconstruction tasks.  
  Within the framework, the input image can be denoted as $\boldsymbol{m}$. The semantic encoder extracts the semantic features of $\boldsymbol{m}$ and maps it to the semantic space by adopting specialized DNN structures. 
  Specifically, consider an RGB image $\boldsymbol{m} \in {\mathbb{R}}^{H \times W \times 3}$ as an example, it is initially split into ${\frac{H}{2}} \times {\frac{W}{2}}$ non-overlapping patches \cite{yang2023witt}, undergoing a series of Swin Transformer \cite{liu2021swin} blocks for feature extraction. Since the number of stages is contingent upon the image resolution, we set two stages according to empirical results. 
  In terms of channel encoder, fully connected (FC) layers are applied to adapt to channel conditions. Thus, a unified source-channel encoding process $\mathcal{F}_{\rm{TX}}(\cdot)$ (parameterized by $\theta$) includes semantic encoding $\mathcal{F}_{\rm{TX}}^{\rm{en}}(\cdot)$ and channel encoding $\mathcal{F}_{\rm{TX}}^{\rm{c}}(\cdot)$ can be given by 
\begin{equation}
    \label{eq: encoding}
    \boldsymbol{x} = \mathcal{F}_{\rm{TX}}^{\rm{c}}(\mathcal{F}_{\rm{TX}}^{\rm{en}}(\boldsymbol{m} ;\theta_{\rm{en}});\theta_{\rm{c}}), 
\end{equation}
 Furthermore, as shown in Fig. \ref{system_model}, the quantization layer $Q_{\rm{TX}}(\cdot)$ further maps $\boldsymbol{x}$ to ${\boldsymbol{\hat x}} \in {\mathbb{Z}}^{B}_{2}$ for subsequent transmissions across multiple physical channels, where $B$ is the fixed numbers of bits for a transmitted image, $\mathbb{Z}_{2} \in \{ 0,1\}$, that is 
 \begin{equation}
    \label{eq: quantization}
    \boldsymbol{\hat x} = Q_{\rm{TX}}(\boldsymbol{x}),
\end{equation}
In this sense, the extracted features are compressed into a sequence of bits with a certain channel bandwidth ratio (CBR). For MNIST dataset,  $\mathrm{CBR} = \frac{B}{8 \times 1\times H \times W}$, while for CIFAR-$10$ dataset, $\mathrm{CBR} = \frac{B}{8 \times 3\times H \times W}$, such a binary quantization method is advantageous for transmission \cite{tong2024alternate,ren_separate_2025}.

  At the \rxs~sides, the transmitted data in semantic BC experiences heterogeneous channel conditions corresponding to different \rxs. Thus, the received signal for \rxn~($n\sim N$) can be denoted as 
\begin{equation}
    \label{eq: received signal} 
    \boldsymbol{Y}_n = {\mathcal {H}_n}(\boldsymbol{\hat x})+\boldsymbol{z}, 
\end{equation}
where ${\mathcal {H}_n}$ is the channel state information (CSI) vector and is usually modeled as an additive white Gaussian noise (AWGN) channel or multiplicative Rayleigh fading channel. Meanwhile, we assume the channel is non-differential, and the training gradient cannot be transmitted from the \tx~to \rxs.
$\boldsymbol{z}$ is noise vector and sampled from a Gaussian distribution, \ie, $\boldsymbol{z}\sim \mathcal{CN}(0, {\sigma ^2}\boldsymbol{I})$.  

Analogously, the dequantization layer $Q_{\rm{RX_n}}(\cdot)$ and decoding process $\FRXN (\cdot)$ (parameterized by $\phi_n$) for \rxn ~involves the reverse steps. To this end, the decoded signal from the $n$-th channel can be given by
\begin{equation} 
    \label{eq: decoded signal} 
    {\boldsymbol U}^{\rm de}_{n} = 
    \mathcal{F}_{{\rm RX}_n}^{\rm{de}}(\mathcal{F}_{{\rm RX}_n}^{\rm{c}}(\boldsymbol{Y}_n;\phi_{n,\rm{c}});\phi_{n,\rm{de}}).
\end{equation}
Afterward, the semantic features ${\boldsymbol U}^{\rm de}_{n}$ are further processed by a light, task-specific module for downstream tasks.   
Regarding the reconstruction task, to overcome the distortion caused by the channel environment, we adopt the mean-squared error (MSE)  between the source and recovered semantic image (i.e., $\boldsymbol{m}^{(t)}$ and $\hat{\boldsymbol{m}}_n ^{(t)}$ as a loss function to learn semantic decoder $\mathcal{F}_{\rm{RX}_n}(\cdot)$, which is denoted as 
 \begin{equation}
    \label{eq: rec-loss}
    \mathcal{L}_{\rm{RX}_n,MSE} ( \phi _n )= \mathbb{E}_{t \sim T}\left [ \left\|{\boldsymbol{m}^{(t)} - \hat{\boldsymbol{m}}_n ^{(t)}}\right\|^2_2\right], 
\end{equation} 
where $t$ denotes the index of the image sample.   

 In terms of the classification task, the semantic decoder $\mathcal{F}_{\rm{RX}_n}(\cdot)$ will be trained by the cross-entropy (CE) loss function, 

 \begin{equation}
    \label{eq: cla-loss}
    \mathcal{L}_{\rm{RX}_n,CE} ( \phi _n )=\frac{1}{T} \sum_{t=1}^{T} -p^{(t)}_n \log q^{(t)}_n
\end{equation} 
where $p^{(t)}_n$ (resp. $q^{(t)}_n$) denotes the ground-truth (resp. predicted) label distribution corresponding to $\hat{\boldsymbol{m}}_n ^{(t)}$ within the $t$-th sample.

\subsection{Problem Formulation} 
\label{sec:problem}

Given the described semantic BC system with one \tx~and $N$ \rxs, our goal is to maximize the task accomplishment performance for \rxn ~by learning the optimal parameters $\left\langle { \{ w_n^ *\}_{n=1}^{N}, {\theta ^ * },\{ \phi _n^ *\}_{n=1}^{N}} \right\rangle$ of the system $\left\langle {{\mathcal {F}_{{\rm{TX}}}},\{{{\mathcal{F} }_{{\rm{R}}{{\rm{X}}_n}}}}\}_{n=1}^N \right\rangle $, that is,
\begin{align}
    \label{eq: optimization problem1}
    &\left\langle {\{ w_n^ *\}_{n=1}^{N}, {\theta ^ * },\{ \phi _n^ *\}_{n=1}^{N}} \right\rangle  \nonumber \\
    &= \mathop {\arg \max }\limits_{\{ w_n\}_{n=1}^{N}, \theta, \{ \phi _n \}_{n=1}^{N}} \sum\nolimits_{n=1}^N \Theta _{ \phi _n}(\boldsymbol{m}, {\hat{\boldsymbol{m}}_n}), 
\end{align}
where $\{ w_n\}_{n=1}^{N}$ represent the learnable weights without known apriori, ${\Theta _{ \phi _n}}$ is the performance metric corresponding to the decoder $n$. For example, for the image reconstruction task, we employ PSNR (peak signal-to-noise ratio) or perceptual metric SSIM (structural similarity index method) as outlined in \cite{tong2024alternate}. For classification tasks, the probability of correct classification is used as the accuracy metric. 

\begin{figure*}[ht]
\begin{subequations}
\begin{align}\tag{8}\label{eq: optimization problem2}
&\min \limits_{\{ w_{n,l}\}_{n=1}^{N}}F(\{ w_{n,l}\}_{n=1}^{N};\theta_l^*, \phi_{n,l}^*)\\  
 & \text{s.t.} \quad \phi_{n,l}^* = \arg\max \limits_{\phi_n} \mathcal{L}_{\rm{RX}_n}(\phi_{n,l};w_{n,l},\theta_l), \qquad \sum\nolimits_{n = 1}^N {{w_n} = 1},w_n \ge 0,\forall n  \label{eq: optimization problem2_a}\\
 & \qquad \theta_l^*(\{ w_{n,l}\}_{n=1}^{N}) = \arg\max_\theta \mathcal{L}_{\rm{TX}}(\theta;\{ w_{n,l}\}_{n=1}^{N}, \phi_{n,l}^*)\label{eq: optimization problem2_b_new}\\
 & \qquad g(\theta_l;\{ w_n\}_{n=1}^{N}, \phi_{n,l}^*) =\mathcal{L}_{\rm{TX}}\left(\theta_l;\{ w_n\}_{n=1}^{N}, \phi_{n,l}^*\right)- \mathcal{L}_{\rm{TX}}\left(\theta_l^*(\{ w_{n,l}\}_{n=1}^{N});\{ w_{n,l}\}_{n=1}^{N}, \phi_{n,l}^*\right) \le 0  \label{eq: optimization problem2_b}
\end{align}
\end{subequations}
\hrulefill
\end{figure*}

\begin{figure*}[ht]
   \begin{equation}
   \label {eq: estimated g(v)}
        \tilde g(\theta_l;\{ w_n\}_{n=1}^{N}, \phi_{n,l}^*)=\mathcal{L}_{\rm{TX}}\left(\theta_l;\{ w_n\}_{n=1}^{N}, \phi_{n,l}^*\right)- \mathcal{L}_{\rm{TX}}\left(\tilde \theta_l^H(\{ w_{n,l}\}_{n=1}^{N});\{ w_{n,l}\}_{n=1}^{N}, \phi_{n,l}^*\right)
   \end{equation}
   \hrulefill    
\end{figure*}
   
\begin{figure}[t!] 
   \centering
   \setlength{\abovecaptionskip}{0.4cm}
		\includegraphics[width=0.98\linewidth]{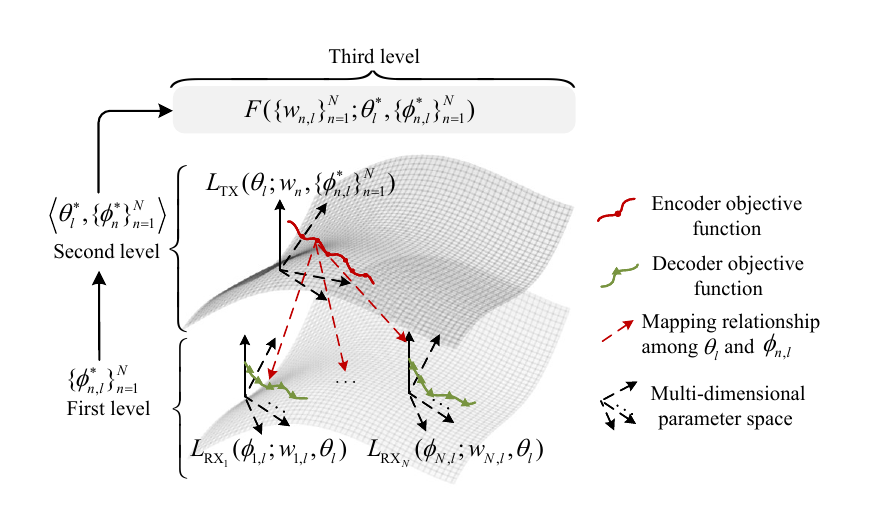}\\ 
		\caption{The parameters relationship between the encoder and decoders.} 
  \label{singleton_condition} 
\end{figure}

  The problem in \eqref{eq: optimization problem1} is challenging to solve by the traditional JSCC framework, due to the hierarchical structure and intricate dependency between the encoder and multiple task-specific decoders. Consequently, it is infeasible to effectively balance a trade-off among the possible conflicting objectives by a simple linear scalarization approach. Meanwhile, the inherent characteristics of multi-task learning naturally lead to a multi-objective optimization problem, often involving inconsistent or even conflicting objectives \cite{sener2018multi}.
  To address these challenges, inspired by the theoretical framework established in our previous work \cite{lu2024self}, we employ an alternating learning-based tri-level mechanism to iteratively train the encoder-decoders pair as well as the weight assignment via RL approach.   

   Specifically, as illustrated in Fig. \ref{singleton_condition}, the problem in \eqref{eq: optimization problem1} is decomposed into a three-level alternating learning framework, where each level is structurally coupled yet allows for independent optimization, enabling flexible adaptation to the heterogeneous tasks of the \rxs.
   Thus, consistent with the singleton condition in \cite{ye2021multi,liu2020generic}, we assume that at the $l$-th iteration ($l \le L$), given a fixed encoder parameter $\theta_l$, there exists a unique optimal solution $\phi_{n,l}^*$ for each decoder, obtained after sufficient local updates. This hierarchical relationship conceptualizes the optimization of decoder and encoder parameters as the first and second levels, respectively, capturing their decoupled yet interdependent roles within the overall learning framework. In this formulation, the first level provides the encoder (\ie, the second level) with the locally optimized decoder parameters $\{\phi_{n,l}^*\}_{n=1}^N$, conditioned on the current encoder parameter $\theta_l$ and the weight configuration $\{w_{n,l}\}_{n=1}^N$. Subsequently, the encoder performs a limited number of local updates via RL approach based on the above optimized decoder parameters $\{\phi_{n,l}^*\}_{n=1}^N$ and the current weight $\{w_{n,l}\}_{n=1}^N$, obtaining an estimated $\theta_l^*$. In this setting, we expect the gap in the \tx~loss before and after the encoder update to be as small as possible, which implies that the current $\theta_l$ is already close to the optimality under the given weight configuration $\{w_{n,l}\}_{n=1}^N$.
   On top of these, a weight assignment module (which will be detailed in Section \ref{sec: Weight Assignment Module}) is integrated at the third level to adaptively optimize the weights for multiple possibly inconsistent or even conflicting objectives at decoder sides. Thus, the encoder leverages this weight assignment module to guide the minimization of its loss. As for the second and third levels, it is assumed that the frequency of parameter updates for $\theta$ and $\{w_{n,l}\}_{n=1}^N$ is consistently upheld. In this way, the third level implicitly steers the joint optimization by dynamically learning the weight assignment $\{w_{n,l}\}_{n=1}^N$, thereby enabling dynamic performance balancing among heterogeneous \rxs.  With this tri-level RL approach, the entire system alternately optimizes the decoders, the encoder, as well as weight assignment module in a bottom-up manner. These steps constitute a single \textit{update cycle}, facilitating iterative learning. 
   
   Finally, this tri-level optimization is formulated as \eqref{eq: optimization problem2}, where the definition of $\mathcal{L}_{\rm{TX}}$ will be provided in Section \ref{sec: The encoder optimization}. 
   In particular, \eqref{eq: optimization problem2_a} corresponds to the first-level local optimization of each decoder, that is, for a given encoder parameter $\theta$, the optimal decoder $n \in \{1,\cdots,N\}$ is approximated by $\kappa$  steps of gradient descent as $ \phi _{n,l}^ {k+1}=\phi _{n,l}^ {k}-\eta \nabla_{\phi_n} \mathcal{L}_{\rm{RX}_n}(\{ \phi _{n,l}^ k\}; \{w_{n,l}\}_{n=1}^N,\theta_l)$, where $\eta$ is the learning rate, $k=0,\dots, \kappa-1$ and $\mathcal{L}_{\rm{RX}_n}$ is the task-specific loss function at \rxn. By \cite{lu2024self}, the decoders are assumed to eventually converge to $\{ \phi _{n,l}^* \}_{n=1}^{N}$ by $\kappa$ iterations on top of given $w_{n,l}$ and $\theta_l$. 
   Similarly, $\theta_l^*(\{ w_{n,l}\}_{n=1}^{N})$ can be attained by \eqref{eq: optimization problem2_b_new}, while \eqref{eq: optimization problem2_b} implies a natural outcome  after the maximization of $\theta_l$.  
   In practice, for simplicity, the optimal parameter $\theta_l^*(\{ w_{n,l}\}_{n=1}^{N})$ can be approximated through an $H$-step gradient descent inspired by a bi-level optimization \cite{ye2024first}. Thus, we obtain an estimate of $ \tilde g(\theta_l;\{ w_n\}_{n=1}^{N}, \phi_{n,l}^*)$ by replacing $\theta_l^*(\{ w_{n,l}\}_{n=1}^{N})$ with $\tilde \theta_l^H (\{ w_{n,l}\}_{n=1}^{N})$, as shown in \eqref{eq: estimated g(v)}.

\section{Implementaions of SemanticBC-TriRL}   
\label{sec: Architecture of the Proposed SemanticBC-TriRL}  
  In this section, we discuss how to optimize the semantic BC framework through the tri-level alternate learning by the weighted PPO algorithm. As shown in Fig. \ref{fig: alternate learning}, we utilize the alternating learning mechanism between codec training and adaptive task prioritization. Hence, at the first level, we highlight the independent calibration of the decoders using a traditional supervised approach. For the second level, the encoder is updated upon the optimized decoders via the PPO algorithm with a given weight assignment. Then, a weight assignment module is employed to determine the optimal weights for the multiple tasks on top of the above two levels. Finally, we also provide the convergence analysis of \texttt{SemanticBC-TriRL}.  

\begin{figure*}[hbt]
   \centering
   \setlength{\abovecaptionskip}{0.1cm}
		\includegraphics[width=0.70\linewidth]{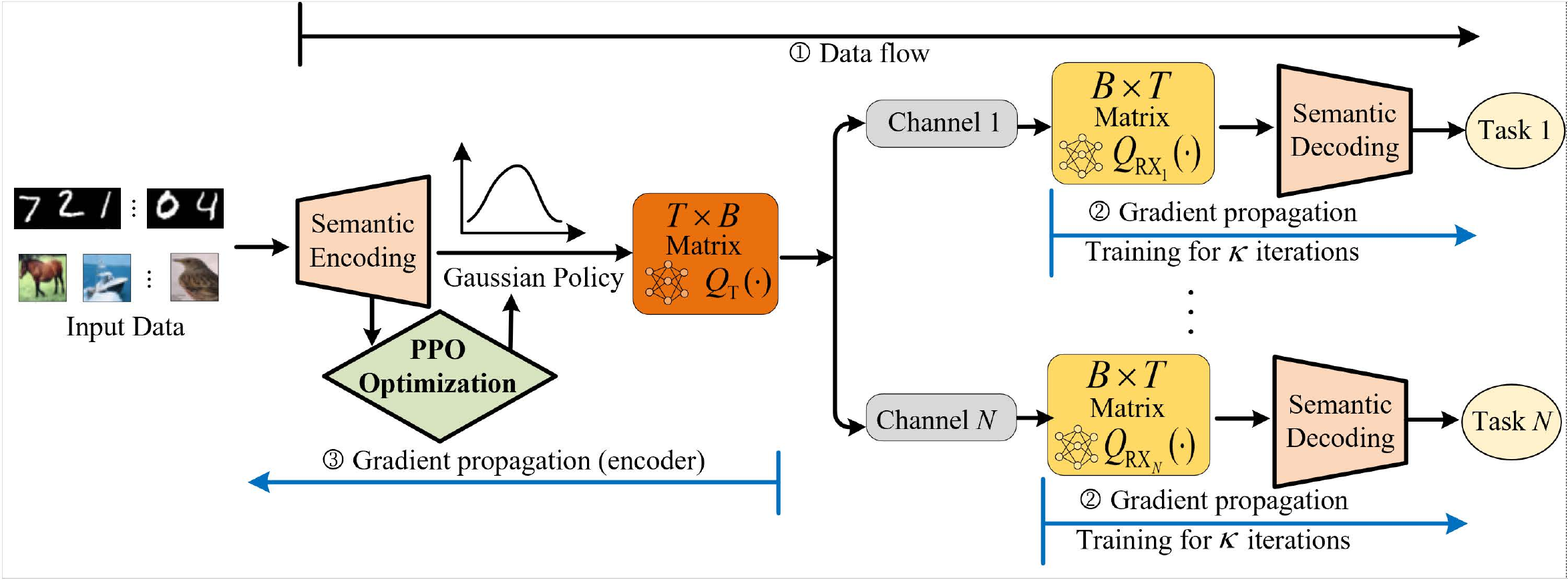}\\ 
\caption{Optimization process under the alternate learning mechanism. \textcircled{1} is the direction of data flow; \textcircled{2} denotes the direction of gradient propagation of \rxn; \textcircled{3} is the direction of gradient propagation of \tx.}
		\label{fig: alternate learning} 
\end{figure*} 

\subsection{Components of Decoder Network} 
\label{sec: The Decoder Optimization}

  Since the downstream tasks have a huge semantic decision space, we adopt simple and efficient traditional supervised learning to better match the performance requirements of various tasks. Meanwhile, unlike the traditional end-to-end training in most semantic communication works, we assume that the encoder and decoders are alternately and independently trained, as outlined in our prior work \cite{lu2024self}. For the decoders with tailored network architectures, which will be enumerated in Section \ref{sec: Simulation Results}, are employed to optimize the loss functions in \eqref{eq: decoded signal} and \eqref{eq: rec-loss}.

\subsection{ Components of the PPO-Driven Encoder Network} 
\label{sec: The encoder optimization}

  Building on the optimized decoders discussed in the previous subsection, as well as the given task-specific weights (which will be determined by the weight assignment module in Section \ref{sec:  Weight Assignment Module}), we now introduce how to effectively train the encoder, so as to better match the performance requirements of various tasks and optimize the objective function at the semantic level. Specifically, we regard the encoder network as an agent, formulate its optimization as a Markov decision process (MDP) setting $(\mathcal{S},\mathcal{A}, P, R,\gamma)$, and leverage a PPO approach to solve it. 

  In particular, the state space $\mathcal{S}$ encompasses the output of semantic encoding, \ie, $s^{(t)}\leftarrow \boldsymbol{x}^{(t)}$, as defined in \eqref{eq: encoding}. The action is  jointly determined by the state $s^{(t)}$ and sample policy $\pi_{\theta}$, parameterized by $\theta$, which can be represented as follows
\begin{equation}
    \label{eq: Gaussian-sample}
   \pi_{\theta} \sim S \left( {\mathcal{N}\left( {\mu (s^{(t)}),\sigma ^2(s^{(t)})} \right)} \right) 
\end{equation} 
where $\mathcal{N}$ denotes a Gaussian distribution and $S$ denote the reparameterized Gaussian sampling strategy with learnable mean value $\mu(s^{(t)})$ and variance $\sigma(s^{(t)})$. 

  Meanwhile, when the encoder takes the action $a^{(t)}\in \mathcal{A}$ with policy $\pi_{\theta}$ under the state $s^{(t)}\in \mathcal{S}$, the state transition function $P(s^{(t+1)}|s^{(t)},a^{(t)})$ is deterministic between two adjacent states. Additionally, for the sake of simplification, we assume the discounted factor $\gamma =1$. Moreover, for a sample $t$ with batch size $T$, the reward function for this learning is assumed to be sparse \cite{tong2024alternate}, as the rewards for the encoder are weighted by the decoders, and decoders can only quantify the corresponding performance after obtaining the final output. To this end, the learning reward for the encoder is denoted as  
  \begin{equation}
  \label{eq: encoder reward} 
      \begin{aligned}
          r^{(t)}_{\pi_\theta}  =  \sum_{n=1}^N w_n {\Theta _{ \phi _n}}(\boldsymbol{m}^{(t)}, {\hat{\boldsymbol{m}}_n^{(t)})})
      \end{aligned}
  \end{equation}  
where ${\Theta _{ \phi _n}}$ is the performance metric corresponding to the decoder $n$.

Upon the above preliminaries of the encoder, the PPO-based encoder learning process is depicted in Fig. \ref{ppo}. Algorithmically, as shown in Fig. \ref{ppo}(a), the actor part leverages two policies, \ie, current policy $\pi_\theta$ and previous policy $\pi_{\theta_{\rm old}}$, in which the parameters of the current policy will be copied to the old policy. In addition, the probability ratio $\mu_t=\frac{\pi_\theta(a^{(t)}|s^{(t)})}{\pi_{\theta_{\mathrm{old}}}(a^{(t)}|s^{(t)})}$ is utilized to prevent large policy update within a short period of time. Furthermore, due to the one-step decision process in our case, the advantage function is simplified as $A^{(t)}_{\pi_\theta} = r_{\pi_\theta}^{(t)}-V(s^{(t)})$. Thus, the loss function of the actor network is represented as 
\begin{align}
   \label{eq: actor loss} 
      &\mathcal{L}_{\rm{actor}}(\theta;\{w_n\}_{n=1}^N, \{\phi_n^*\}_{n=1}^N) \nonumber\\
      &=- \mathbb{E}_{t \sim T} \left[\min \left(\mu_t A^{(t)}_{\pi_\theta}, \operatorname{clip}\left(\mu_t, 1-\epsilon, 1+\epsilon\right) A^{(t)}_{\pi_\theta}\right)\right]  
\end{align}
where the function $\operatorname{clip}(\cdot,1-\epsilon,1+\epsilon)$ is designed to prevent excessive policy updates by constraining the policy ratio within the range $[1-\epsilon, 1+\epsilon]$, where $\epsilon$ is a hyperparameter. This mechanism penalizes abrupt policy changes, ensuring stable optimization and preventing drastic shifts in the policy. 

As for the critic network, as illustrated in Fig. \ref{ppo}, along with the actor network is updated, the value function $V_{\chi}(s^{(t})$ is estimated by another critic network using the MSE loss function 
\begin{equation}
  \label{eq: value loss} 
      \mathcal{L}_{\rm value}(\chi) =\mathbb{E}_{t \sim T}\left[ \left(V_{\chi}(s^{(t)})- V^{\rm targ}(s^{(t)})\right)^2 \right]
  \end{equation}  
where the target value $V^{\rm targ}(s^{(t)})$ is simplified to the immediate reward $r_{\pi_\theta}^{(t)}$ since there is only one decision step with meaningful rewards. 

\begin{figure*}[t!]
  \begin{equation}
    \label{eq: TX loss}
     \mathcal{L}_{\rm{TX}}\left(\theta;\{w_n\}_{n=1}^N, \{\phi_n^*\}_{n=1}^N\right)= \mathcal{L}_{\rm{actor}}\left(\theta;\{w_n\}_{n=1}^N, \{\phi_n^*\}_{n=1}^N\right) + \sum_{n=1}^{N} w_n\mathcal{L}_{\rm{RX}_n}\left(\theta;\{w_n\}_{n=1}^N, \{\phi_n^*\}_{n=1}^N\right)    
\end{equation}
\hrulefill 
\end{figure*}

\begin{figure*}[hbt]
   \centering
   \setlength{\abovecaptionskip}{0.3cm}
		\includegraphics[width=0.85\linewidth]{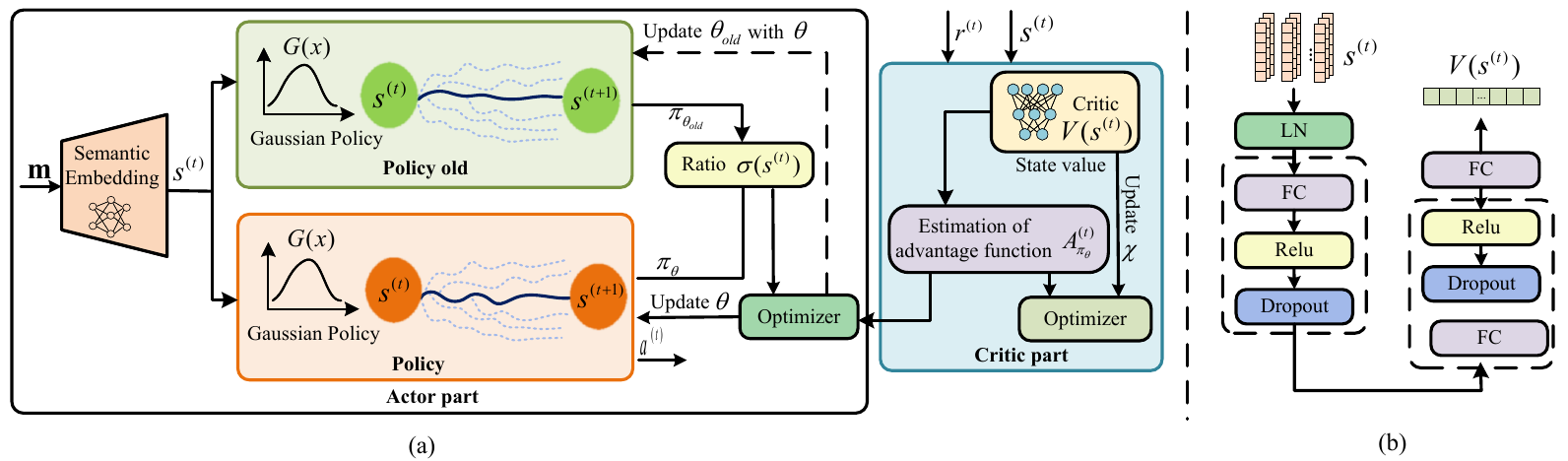}\\ 
\caption{The encoder learning process via weighted PPO algorithm. }
		\label{ppo} 
\end{figure*}

Moreover, to address the potential instability associated with training the encoder solely based on PPO’s policy loss, and to accelerate convergence and enhance the overall performance of our Semantic BC framework, we introduce an additional auxiliary decoder-side supervised loss. This additional loss provides explicit and stable task-specific supervision and serves to rectify the potentially suboptimal update direction induced by PPO. Thus, the overall encoder loss function is formulated as \eqref{eq: TX loss}.

\subsection{Optimization of SemanticBC-TriRL with Weight Assignment Module}     
\label{sec: Weight Assignment Module}

 In this section, building upon the definitions of the encoder and decoders, we introduce a weight assignment module to dynamically balance the task weights of encoder learning in \eqref{eq: optimization problem2} under a constrained function. 
 Before delving further, we first provide a more detailed clarification of problem \eqref{eq: optimization problem2}. Specially, at the first level, as a default setting in each iteration, given the encoder parameter $\theta_l$, the decoder parameters $\phi_n^*$ are iteratively optimized after $\kappa$ steps of local updates.
 Consequently, to simplify the notation, at $l$-th iteration, we define joint optimization variable as $v_l \equiv (\{ w_{n,l}\}_{n=1}^{N},\theta_l)$. Thus, in \eqref{eq: optimization problem2}, the third level $F(v_l)=\mathcal{L}_{\rm{TX}}(v_l;\{\phi_{n,l}^*\}_{n=1}^N)$ governs the joint optimization of the encoder parameters $\theta_l$ and the weight assignment module $\{w_n\}_{n=1}^N$, while the second level function $\tilde g(v_l)$ enforces the corresponding constraints. In this sense, 
 the encoder adaptively updates the shared encoder parameters, thereby alleviating potential conflicts across multiple tasks. Meanwhile, the second level constraint function $\tilde g(v_l)$ aims to approximate the optimal encoder parameter $\theta_l$ by performing $H$ steps of local gradient descent on the encoder parameters. 
 Intuitively, the learning rate of $F(v_l)$ should be properly controlled by constraint function $\tilde g(v_l)$, so that the complicated relationship between $\theta$ and $w_n$ can be accordingly addressed to prevent the excessive sacrifice of certain tasks during the optimization process, ensuring a more balanced and equitable improvement. 
 In this way, at $l$-th iteration, we assume $v_{l+1} = v_l + \eta d_l$, where $d_l$ is the update direction of $v_l$. Thus, we aim to find $d_l$ to simultaneously minimize third level objective function $F(v_l)$ and second level constraint function $\tilde g(v_l)$. 
 
 Furthermore, to reformulate problem \eqref{eq: optimization problem2} into a tractable optimization framework and ensure that all objectives are decreased as much as possible at each iteration, we follow the approach inspired by \cite{gong2021bi, ye2024first}.  In particular, to decrease $F(v_l)$ at $l$-th iteration, we consider its first-order Taylor approximation $F(v_l+\eta d_l) \approx F(v_l)+ \eta \langle \nabla F(v_l), d_l \rangle$. Hence, minimizing $F(v_l + \eta d_l)$ is equivalent to minimizing $\langle \nabla F(v_l), d_l \rangle$. Analogously, to decrease  $\tilde{g}(v_l)$, we have the same approximation $\tilde g(v_l+\eta d_l) \approx \tilde{g}(v_l)+ \eta \langle \nabla \tilde{g}(v_l), d_l \rangle$. To ensure $\tilde g(v_l+\eta d_l) < \tilde{g}(v_l)$, we impose the constraint $\langle \nabla \tilde{g}(v_l), d_l \rangle \le -\rho_l$, where $\rho_l$ is a non-negative control barrier and should be strictly positive $\rho_l>0$ \cite{gong2021bi, ye2024first}. Here, to ensure stricter constraint and the update meaningfully improves feasibility, we set $\rho_l= \beta\|\nabla\tilde{g}(v_l)\|^2$ by default \cite{gong2021bi}.
 In this regard, the problem \eqref{eq: optimization problem2} can be addressed to find a descent direction $d_l$ for $v_l=[\{w_{n,l}\}_{n=1}^N,\theta_l]$ by solving the following problem 
\begin{align}
\label{eq: optimization problem3}
d_l=\arg\min_{d}\langle \nabla F(v_{l}),d\rangle +\frac{1}{2}\|d\|^{2} \ \text{s.t.} \ {\langle\nabla\tilde{g}(v_l)},d \rangle \le -\rho_l
\end{align}
where $\frac{1}{2}\|d\|^{2}$ is a regulation term. 
In this way, $F(v_l)$ is minimized while the constraint function $\tilde{g}(v)$ is progressively satisfied. To solve this constrained optimization, we adopt the Lagrange multiplier method, yielding:
\begin{align}
\label{eq: Lagrange function}
   &\mathcal{L}(d,\lambda_l) \nonumber\\
   &=\langle \nabla F(v_{l}),d\rangle+\frac{1}{2}\|d\|^2+\lambda_l\left(\nabla\tilde{g}(v_l)^\top d+\beta\|\nabla\tilde{g}(v_l)\|^2\right) 
\end{align}
where $\lambda_l$ is the Lagrange multiplier and satisfies $\lambda_l >0$. 
Furthermore, the problem formulated in \eqref{eq: Lagrange function} yields a closed form solution: $d_l=-\nabla F(v_l)-\lambda_l\nabla\tilde{g}(v_l)$ and $\lambda_l\ge \frac{\rho_l-\nabla \tilde g(v_l)^\top\nabla F(v_l)}{\|\nabla \tilde g(v_l)\|^2}=\max(\frac{\rho_l-\nabla \tilde g(v_l)^\top\nabla F(v_l)}{\|\nabla \tilde g(v_l)\|^2},0)$. Since we consider the expectation of the encoder loss over a batch in \eqref{eq: TX loss}, it is assumed that the clipping of a single sample does not impact the overall gradient computation.  
Thus, at $l$-th iteration, $F(v_l) = \mathcal{L}_{\rm{TX}}(v_l;\{\phi_{n,l}^*\}_{n=1}^N))$ and $\nabla_v F(v_l)=\left[\nabla_{\boldsymbol{w}} F(v_l), \nabla_{\theta} F(v_l)\right]$, where $\boldsymbol{w}=\{w_n\}_{n=1}^N \in \mathbb{R}^N$. Thus, we have \eqref{eq: gradient of w} and \eqref{eq: gradient of theta}.  
Analogously, the gradients of $\nabla \tilde g(v_l)$ is obtained by \eqref{eq: gradient of g(v)}. 
In this sense, the computed update direction $d_l$ can be leveraged to derive the gradient with respect to the $\{w_{n,l}\}_{n=1}^N$ and $\theta_l$, which are then used to update these parameters accordingly.

\begin{figure*}[ht]
\begin{subequations}
\begin{align}
    &\nabla_{\boldsymbol{w}} F(v_l) =\nabla_{\boldsymbol{w}} \mathcal{L}_{\rm{TX}}\left(v_l;\{\phi_{n,l}^*\}_{n=1}^N\right) = -E_{t \sim T}\left[\frac{\pi_{\theta_l}(a^{(t)})|s^{(t)})}{\pi_{\theta_{l,\mathrm{old}}}(a^{(t)}|s^{(t)})}\left(
\begin{array}
{c}\Theta_{\phi_{1,l}^{\kappa}}(\boldsymbol{m},\hat{\boldsymbol{m}}_{1}^{(t)})\\
\cdots \\ 
\Theta_{\phi_{N,l}^{\kappa}}(\boldsymbol{m},\hat{\boldsymbol{m}}_{n}^{(t)})
\end{array}\right)\right] + \left(
\begin{array}
{c}\nabla_{w_1}\mathcal{L}_{\rm{RX}_1}(\theta;w_1,\phi_1^*)\\
\cdots \\ 
\nabla_{w_N}\mathcal{L}_{\rm{RX}_N}(\theta;w_N,\phi_N^*)
\end{array}\right)\label{eq: gradient of w}\\
& \nabla_{\theta} F(v_l)=\nabla_{\theta} \mathcal{L}_{\rm{TX}}\left(v_l;\{\phi_{n,l}^*\}_{n=1}^N\right) \nonumber \\
& = -E_{t \sim T}\left[\frac{\nabla\pi_{\theta_l}(a^{(t)})|s^{(t)})}{\pi_{\theta_{l,\mathrm{old} }}(a^{(t)}|s^{(t)})}\left(\sum_{n=1}^N w_{n,l} \Theta _{ \phi_{n,l}^{\kappa}}(\boldsymbol{m},{{\hat{\boldsymbol m}}_n^{(t)}}) -V(s^{(t)})\right)\right] +  \sum_{n=1}^{N} w_n \nabla_\theta \mathcal{L}_{\rm{RX}_n} \left(\theta;\{w_n\}_{n=1}^N, \{\phi_n^*\}_{n=1}^N\right)
\label{eq: gradient of theta}
\end{align} 
\end{subequations}
\hrulefill
\end{figure*}

\begin{figure*}
\begin{equation}
     \label{eq: gradient of g(v)}
     \nabla \tilde g(v_l)= \left[\nabla_v F(v_l) - \nabla_{\boldsymbol{w}}\mathcal{L}_{\rm{TX}}\left(\tilde \theta_l^H;\{ w_{n,l}\}_{n=1}^{N}, \phi_{n,l}^*\right) \right],
\end{equation}
\hrulefill
\end{figure*}

\noindent \textbf{Remark:} Thus, building upon the decoders' local updates and PPO-driven optimization of the encoder side, the proposed \texttt{SemanticBC-TriRL}, facilitated by alternate learning within a tri-level structure, enables adaptive semantic encoding to avoid possible conflicts among multiple tasks in a bottom-up manner. This adaptability is achieved by dynamically adjusting the weight assignment to different decoders, modeled as a quadratic programming problem in \eqref{eq: optimization problem3}. In particular, as illustrated in Fig. \ref{fig: alternate learning}, in a \textit{update cycle} at $l$-th iteration, the \rxs~ locally update their parameters for $\kappa$ iterations with fixed encoder parameter $\theta_l$.  
Further expand the optimization process of \tx~at the second level, \tx~ updates once and its optimization aims to satisfy the constraints function $g(v_l)$ with given decoders' parameters $\{\phi _{n,l}^ \kappa \}_{n=1}^{N}$ while the third level focuses on learning the proper weight $\{w_{n,l}\}_{n=1}^N$ to mitigate possible conflicts by utilizing the optimized $\{\phi _{n,l}^ * \}_{n=1}^{N}$. In this way, by introducing a straightforward weight assignment module, a trade-off solution can be achieved with relatively low computational cost. 
Finally, the detailed procedures for training \texttt{SemanticBC-TriRL} are summarized in \textbf{Algorithm } \ref{al: algorithm1}.

\begin{algorithm}[t]
   	\renewcommand{\algorithmiccomment}{\textbf{// }}
	\renewcommand{\algorithmicrequire}{\textbf{Input:}}
	\renewcommand{\algorithmicensure}{\textbf{Output:}}
   \caption{The training process of \texttt{SemanticBC-TriRL}.}   
   \label{al: algorithm1}
\begin{algorithmic}[1]
                \REQUIRE Learning rate ($\gamma$, $\eta$, $\varsigma$), number of iterations ($H, \kappa$),  input sequence ${\boldsymbol{m}}$, number of \rxs ~$N$. 
               \ENSURE Encoder parameter $\theta$, decoder parameter $\phi_n$. \\
               \FOR{epoch=$1: E_e$}  
                \FOR {$l=1:L$}
                \STATE $\tilde \theta_l \leftarrow \theta_l$;  
                \FOR {$n=1:N$}
               \FOR {$k=1:\kappa$}
               \STATE Sample a batch of data and update ${\phi _{n,l}^{k+1}}\leftarrow {\phi _{n,l}^k}- \gamma{ \nabla_{\phi _n^k}{L_{\rm{RX}_n}}({\phi _{n,l}^k};\{w_{n,l}\}_{n=1}^N,\theta_l) }$; 
                \ENDFOR\\
               \ENDFOR\\
               \FOR {$h=1:H$}
               \STATE Substituting $\mu_{l,t}=\frac{\pi_{\tilde \theta_l^h}(a^{(t)}|s^{(t)})}{\pi_{\tilde \theta_{l,\mathrm{old}}^h}(a^{(t)}|s^{(t)})}$ and $A^{(t)}_{\pi_{\theta_l^h}} = r_{\pi_{\theta_l^h}}^{(t)}-V(s^{(t)})$ into \eqref{eq: TX loss} yields $\mathcal{L}_{\rm TX}\left(\tilde \theta_l^h;\{w_{n,l}\}_{n=1}^N,{\phi _{n,l}^\kappa}\right)$;\\
               \STATE $\tilde \theta_l^{h+1} \leftarrow \tilde \theta_l^{h} - \gamma \nabla_{\theta} \mathcal{L}_{\rm TX}\left(\tilde \theta_l^h;\{w_{n,l}\}_{n=1}^N,{\phi _{n,l}^\kappa}\right)$; 
               \ENDFOR\\ 
               \STATE Compute $\tilde g(v_l)$ by \eqref{eq: estimated g(v)} ;\\
               \STATE Compute $\nabla \tilde g(v_l)$ by \eqref{eq: gradient of g(v)};\\
               \STATE Compute $\mathcal{L}_{\rm TX}\left(\theta_l;\{w_{n,l}\}_{n=1}^N,{\phi _{n,l}^\kappa}\right)$ by \eqref{eq: TX loss};\\
               \STATE Compute $\nabla F(v_l)$ by \eqref{eq: gradient of w} and \eqref{eq: gradient of theta};\\
               \STATE Compute $\lambda_l = \max(\frac{\rho_l-\nabla \tilde g(v_l)^\top\nabla F(v_l)}{\|\nabla \tilde g(v_l)\|^2},0)$ by \eqref{eq: Lagrange function};\\
               \STATE Set $d_l=-\nabla F(v_l)-\lambda_l\nabla\tilde{g}(v_l)$; \\
               \STATE Update $v_{+1}=[\{w_{n,l}\}_{n=1}^N, \theta_l]$ by $v_{l+1} \leftarrow v_l+ \eta d_l$; \\
               \STATE Update value function $\chi^{l+1} \leftarrow \chi_l - \varsigma \nabla L_{value}(\chi)$;\\
               \STATE $\theta_{l, \rm{old}} \leftarrow \theta_l$;\\
               \ENDFOR\\
		\ENDFOR\\
  \STATE 	\COMMENT{\textbf{Finish training}} 
  \STATE Return $v=[\{w_{n}\}_{n=1}^N, \theta]$ and $\phi_n$. 	
\end{algorithmic}  
\end{algorithm}

\subsection{Convergence Analysis}

In this section, we provide the convergence analysis of \texttt{SemanticBC-TriRL}.   
Beforehand, in order to better formulate the optimization between the encoder and the weight assignment module, we make the following assumptions.

\begin{assumption} [Smoothness]
    \label{assump: assumption1} 
 (1) For any given $v_l \equiv (\{ w_{n,l}\}_{n=1}^{N},\theta_l)$, $f(v_l)=\mathcal{L}_{\rm{TX}}\left(v_l; \phi_{n,l}^*\right)$, where $\theta_l$ is treated as the optimization variable, satisfies $c$-strongly convex with respect to $\theta_l$ (Assumption $2$ in \cite{ye2024first}). Given that both $\mu^{(t)}$ and $A^{(t)}=r_{\pi_\theta}^{(t)}-V(s^{(t)})$ are continuous and bounded in \eqref{eq: TX loss}, $\nabla f(v_l)$ is $L_f-$Lipschitz continuous, namely $\|\nabla f(v_1)-\nabla f(v_2)\| \le L_f\| v_1-v_2\|$.
 (2) The differentiable function $F(v_l) = \mathcal{L}_{\rm{TX}}(v_l;\{\phi_{n,l}^*\}_{n=1}^N))$, where $\{ w_{n,l}\}_{n=1}^{N}$ is treated as the optimization variable, also has an $L_F-$Lipschitz continuous gradient $\nabla F(v_l)$, namely $\|\nabla F(v_1)-\nabla F(v_2)\| \le L_F\| v_1-v_2\|$.

\end{assumption}  
\begin{assumption} [Boundedness of \cite{liu2022bome}] 
    \label{assump: assumption2} 
There exists a constant $M<\infty$, such that the $L_2$ norm of $\nabla F(v_l)$ is upper-bounded by $M$. 
\end{assumption} 

Notably, Assumption \ref{assump: assumption1} and Assumption \ref{assump: assumption2} are both widely used in standard optimization \cite{liu2022bome, ye2024first}. In our setting,  beyond the locally optimized decoder parameters $\phi_n^*$ at the first level, the key aspect lies in the application of optimization framework with the shared encoder and the weight assignment module. The learning processes associated with the third level $F(v_l)$ and the second level of $f(v_l)$ are naturally differentiable and bounded. 
Moreover, since the elements of the Hessian matrix of $\nabla f(v_l)$ are bounded, the condition for Lipschitz continuity is also satisfied. Assumption \ref{assump: assumption2} aims to ensure a stable update in a controlled manner.

We expect our proposed \texttt{SemanticBC-TriRL} can find Karush-Kuhn-Tucker (KKT) stationary point to solve the optimization problem \eqref{eq: optimization problem3}. We first give the KKT condition of \eqref{eq: optimization problem3} as Lemma \ref{lem: lemma1}. 
\begin{lemma}[KKT condition of \texttt{SemanticBC-TriRL}]
\label{lem: lemma1}
For $\boldsymbol{w} \in \Delta^{N-1}$ and $\lambda_l >0$, at a local optimal point $v^*$, the following three conditions hold. 
(1) Stationarity: $-\nabla F(v^*)-\lambda_l\nabla\tilde{g}(v^*)=0$; (2)
Primal feasibility: $g(v^*) \le 0$; (3) Complementary slackness: $\lambda_l \cdot \left( \nabla g(v^*)^\top d - \rho_l \right) = 0$.
\end{lemma} 
Since the constraint function $g(v)$ of the encoder is ill-posed and does not have an interior point \cite{ye2024first}. Inspired by \cite{ye2024first}, we only consider the stationarity condition and primal feasibility condition to converge to a weak stationarity point. Specially, at $l$-th iteration, the stationarity condition is measured by $\psi(v_l)=\|-\nabla F(v_l)-\lambda_l\nabla{g}(v_l)\|^2$. Meanwhile, to satisfy the primal feasibility condition $g(v_l) \le 0$, $g(v_l)$ should decrease to $0$ as possible. Therefore, to analyze the convergence of our proposed \texttt{SemanticBC-TriRL}, we should focus on the convergence rates of $\psi(v_l)$ and $g(v_l)$ as $\mathcal{L}_{\rm TX}$ increases. We first have the following Lemma. 

\begin{lemma}[Proof of Theorem 1 of \cite{liu2022bome}] 
    \label{lem: lemma2}
    Under Assumption \ref{assump: assumption1}, there exists a constant $C$ and $H>C>0$, such that $\| \nabla\tilde g(v)\| \ge C_a \|\nabla g(v)\|$, where $C_a=\frac{1}{2}$. 
\end{lemma}
Lemma \ref{lem: lemma2} implies that the approximated $\|\nabla \tilde g(v)\|$ is usually small than $\|\nabla g(v)\|$. 
Moreover, we have the upper bound of $g(v_l)$ and $\sum_{l=0}^{L-1}||\nabla g(v_l)||^2$ by the following Lemma. 

\begin{lemma}[Lemma 10 in \cite{liu2022bome}, Lemma B.2(a) and Lemma Lemma B.2(d) of \cite{ye2024first}] 
    \label{lem: lemma3}
    Under assumption \ref{assump: assumption1}, suppose $g(v_0) \le B$ ($B$ is a positive constant), we have $g(v_l) \le \Gamma_1(l)g(v_0) + \Delta \le \Gamma_1(l)B + \Delta$ and $\sum_{l=0}^{L-1}||\nabla \tilde g(v_l)||^2 \le \mathcal{O}\left(\frac{1}{\eta}+ L \Delta \right)$, where $\Delta=\mathcal{O}(\Gamma(H))+ \eta$, and $\Gamma(H)$ (resp. $\Gamma_1(l)$) is an exponential decay with respect to $H$ (resp. $l$). 
\end{lemma}

Based on the above definitions, it is ready to analyze the convergent rate 
of our encoder optimization iterated by alternate learning, as shown in Theorem \ref{thm: theorem1}. 
\begin{theorem}
    \label{thm: theorem1}
    Under Assumption \ref{assump: assumption1} and Assumption \ref{assump: assumption2}, if $\nabla g(v_0)\le B$ ($B$ is a positive constant), there exists a $C>0$, when $H>C$, for any $L>0$, we have 
    \begin{align}
        \label{eq: theorem1}
    \max \limits_{l\le L}\left\{ \min \limits_{l\le L}\psi(v_l), g(v_l)\right\}=\mathcal{O}\left( \Gamma\left(\frac{H}{2}\right) +\sqrt{\eta} + \sqrt{\frac{1}{L\eta}} \right)
    \end{align}
\end{theorem}
\begin{proof}

 Our goal is to measure the convergent rate from the stationary condition $\psi(v_l)=\|-\nabla  F(v_l)-\lambda_l\nabla{g}(v_l)\|^2$ and primal feasibility condition $g(v_l)$. For the stationary condition $\psi(v_l)$, we have \eqref{eq: stationary function}. Specifically, inequality (a) comes from the arithmetic mean-geometric mean (AM-GM) inequality, and for inequality (b), the first term is due to the Assumption \ref{assump: assumption2}. In terms of inequality (c), we substitute $\lambda_l\ge \max(\frac{\rho_l-\nabla \tilde g(v_l)^\top\nabla F(v_l)}{\|\nabla \tilde g(v_l)\|^2}, 0)$ and $\rho_l = \beta\|\nabla\tilde{g}(v_l)\|^2$ to it and obtain corresponding results. 
 Then,  inequality (d) can be obtained by applying the basic inequality $-\nabla \tilde g(v_l)\nabla F(v_l) \le \|\nabla \tilde g(v_l)\| \cdot \|\nabla F(v_l)\|$.  
  Finally, inequality (e) and (f) utilize the $\| \nabla\tilde g(v_l)\| \ge \frac{1}{2} \|\nabla g(v_l)\|$ in Lemma \ref{lem: lemma2} to derive the simplification. 

 \begin{figure*}[ht] 
     \begin{align}
\label{eq: stationary function}
  \psi(v_l)&=\left \|-\nabla F(v_l)-\lambda_l\nabla{g}(v_l)\right\|^2 \nonumber\\
   & \mathop  \le \limits^{(a)} 2\|\nabla F(v_l)\|^2 + 2\lambda_l^2 \left\| \nabla  g(v_l)\right\|^2 \nonumber\\
   & \mathop  \le \limits^{(b)} 2 M^2 + 2  \lambda_l^2 \| \nabla  g(v_l)\|^2 \nonumber\\
   &  \mathop  \le \limits^{(c)}2 M^2 + 2\left(\frac{\rho_l - \nabla \tilde g(v_l)^\top \nabla F(v_l) }{\| \nabla \tilde g(v_l)\|^2} \right)^2 \| \nabla  g(v_l)\|^2 \nonumber \\
   &  \mathop  \le \limits^{(d)} 2 M^2 + 2 \left(  \frac{\beta^2 \| \nabla \tilde g(v_l)\|^4 + \| \nabla \tilde g(v_l)\|^2 \|\nabla F(v_l)\|^2 + 2 \beta \|\nabla \tilde g(v_l)\|^2 \|\nabla \tilde g(v_l)\| \|\nabla F(v_l)\|}{\|\nabla \tilde g(v_l)\|^4}\right) \| \nabla  g(v_l)\|^2 \nonumber \\
   & \mathop  \le \limits^{(e)} 10 M^2 + 2\beta ^2 \|\nabla g(v_l)\|^2 +16 \beta M \|\nabla \tilde g(v_l) \|\nonumber \\ 
   & \mathop  \le \limits^{(f)} 10 M^2 + 8\beta ^2 \|\nabla \tilde g(v_l)\|^2 + 16 \beta M \|\nabla \tilde g(v_l) \|
\end{align} 
\hrulefill
 \end{figure*} 

For the second term of the right-hand side of \eqref{eq: stationary function}, according to Lemma \ref{lem: lemma3}, we have 
\begin{align}
    \label{eq: stationary function1 }
    \sum \limits_{l=0}^{L-1} \|  \nabla \tilde g(v_l)\|^2 &\le  \mathcal{O}\left(\frac{1}{\eta}+ L \Delta\right)\nonumber\\
    & =  \mathcal{O}\left( \frac{1}{\eta} + L\Gamma(H)+ L\eta \right)
\end{align}
Meanwhile, for the third term of \eqref{eq: stationary function}, combining Lemma \ref{lem: lemma3} with the Cauchy-Schwarz inequality, we obtain that 
\begin{align}
    \label{eq: stationary function2}
     \sum \limits_{l=0}^{L-1} \|\nabla \tilde g(v_l) \| &\le \sqrt{L} \left(\sum \limits_{l=0}^{L-1} \|\nabla \tilde g(v_l) \|^2\right)^{1/2} \nonumber\\
     &=\sqrt{L} \left(\mathcal{O}\left(\frac{1}{\eta}+L \Gamma(H)+ L\eta\right)\right)^{1/2}\nonumber\\
     &=\mathcal{O}\left(\sqrt{\frac{L}{\eta}}+L\Gamma\left(\frac{H}{2}\right)+L\sqrt{\eta}\right)  
\end{align} 
Along with \eqref{eq: stationary function}, we have  
\begin{align} 
    \label{eq: stationary function3} 
\min \limits_{l<L}\psi(v_l) \le \frac{1}{L}  \sum \limits_{l=0}^{L-1} \psi(v_l) =\mathcal{O} \left( \sqrt{\frac{1}{L\eta}} + \Gamma\left(\frac{H}{2}\right) +\sqrt{\eta} \right)
\end{align}
According to Lemma \ref{lem: lemma3}, we have $g(v_l)= \mathcal{O} \left ( \Gamma_1(l)B + \Gamma(H) + \eta \right)$. Finally, we obtain \eqref{eq: theorem1}.  
\end{proof}  

\noindent \textbf{Remark:} From Theorem \ref{thm: theorem1}, in alignment with stationary condition $\psi(v_l)$ and feasibility condition $g(v_l)$, \eqref{eq: theorem1} implies that the convergent rate will rapidly improve as $L$ and $H$ increase. As for \eqref{eq: theorem1}, inspired by \cite{liu2022bome}, when $\eta = \mathcal{O}(L^{-1/2})$, the convergent rate can be denoted as $\max \limits_{l\le L}\left\{ \min \limits_{l\le L}\psi(v_l), g(v_l)\right\}=\mathcal{O}\left( \Gamma\left(\frac{H}{2}\right)+L^{-1/4} \right)$.

\section{Simulation Results}
\label {sec: Simulation Results}

\subsection{Simulation Settings}
\label{sec: Simulation Settings}

In this section, we validate the superiority of \texttt{SemanticBC-TriRL} on two typical tasks, including both image reconstruction and image classification tasks. In terms of datasets, we choose the popular MNIST and CIFAR-10 datasets for training and testing, which are introduced as follows.
\begin{itemize}
    \item \textbf{MNIST}: MNIST\footnote{\url{http://yann.lecun.com/exdb/mnist}} is a classic dataset containing $70,000$ grayscale images of handwritten digits ranging from ``$0$'' to ``$9$''), each with a resolution of $ 28\times28$ pixels. The dataset is divided into a training set of $60,000$ samples and a testing set of $10,000$ samples.
    \item \textbf{CIFAR-$10$}: CIFAR-$10$ dataset\footnote{\url{https://www.cs.toronto.edu/~kriz/cifar.html}} containing $60, 000$ RGB images with the fixed size of $32\times32\times 3$ across $10$ distinct classes (\ie, airplane, automobile, bird). It is divided into a training set of $50,000$ images and a testing set of $10,000$ images. 
\end{itemize}

The other related parameters are listed in Table \ref{default parameters}.
\begin{table}[t]
		\centering
		\caption{The default parameters in \texttt{SemanticBC-TriRL}}      
		\label{default parameters}		  
		\def\arraystretch{1.0}  
    \begin{tabular}{m{1cm} m{4cm}m{2cm} } 	
			\hline
			Parameters &Description& Value\\  
                \hline
                $T$&Batch size& $64$ \\ 
                $\eta $&Learning rate of encoder and decoders& $1e-3$\\
                $\kappa$& Local iterations for decoders in each \textit{update cycle}&$100$\\ 
                $H$&No. of inner steps at \tx &$5$\\   
                $E_e$ &No. of training epochs &$71$\\
                $B$& The fixed bits for image transmission & $128$, $1024$, $5000$\\
                $\mathrm{CBR}$&Channel bandwidth ratio &$0.02$, $0.16$, $0.2$\\  
                \hline 
		\end{tabular}   
        \vspace{-0.35cm}    
\end{table}

\textbf{Evaluation Metrics:} For the image reconstruction task, we evaluate the performance of image transmission by widely used pixel-wise PSNR and perceptual metric SSIM. 
 Specifically, PSNR measures the ratio between the maximum possible power of the signal and the power of the noise that corrupts the signal \cite{bourtsoulatze2019deep}, which is defined as 
\begin{equation}
    \label{eq: psnr}
     \Theta_\text{PSNR} = 10\log_{10}\left(\frac{\text{MAX}^2}{\text{MSE}}\right)\quad(\text{dB})
\end{equation}
where $\text{MSE} = d(\boldsymbol{m}, \boldsymbol{\hat{m}_n})$ denotes the mean squared-error, and MAX is the maximum possible value of the image pixels, \ie, $255$ for $24$ bit depth RGB images \cite{tong2024alternate}. On the other hand, the SSIM score leverages the inter-dependencies of pixels and reveals the perception information to some extent. In this way, $\Theta_\text{SSIM}(\boldsymbol{m}, \boldsymbol{\hat{m}_n})$ compares $\boldsymbol{m}$ and $\boldsymbol{\hat{m}_n}$ in terms of the  \emph{luminance}, \emph{contrast} and \emph{structure}, and can be computed as 
\begin{align}
    \label{eq: ssim}
     &\Theta_\text{SSIM}(\boldsymbol{m}, \boldsymbol{\hat{m}_n}) \nonumber\\
     &= {(\rho_\text{l}(\boldsymbol{m}, \boldsymbol{\hat{m}_n})^{\lambda_1} } \cdot {(\rho_\text{c}(\boldsymbol{m}, \boldsymbol{\hat{m}_n})^{\lambda_2} } \cdot {(\rho_\text{s}(\boldsymbol{m},  \boldsymbol{\hat{m}_n}))^{\lambda_3} } \in [0,1], 
\end{align}
where the functions $\rho_\text{l}$, $\rho_\text{c}$ and $\rho_\text{s}$ reflect the perceived changes in terms of luminance, contrast, and structural patterns between two images. Meanwhile, $\lambda_1$, $\lambda_2$ and $\lambda_3$ denote the exponential coefficient \cite{lu2022semantics}. 

In addition, the classification accuracy is defined as the ratio of the number of correctly classified samples to the total number of samples. Besides, it is worth noting that we adopt the perceptual metric SSIM \cite{tong2024alternate} as the learning reward in  \eqref{eq: encoder reward} for data reconstruction, while for classification tasks, the classification accuracy is taken as a learning reward. Also, we randomly crop and flip the images for data augmentation.

\textbf{DNN structures}: As shown in Fig. \ref{system_model}, the number of \rxs~ is set to $N=2$ and the corresponding DNNs for our semantic BC model are specially designed for different datasets, which are summarized in Table \ref{tab: mnist-structure} and Table \ref{tab: cifar10-structure}. In particular, in terms of MNIST dataset, the default fixed length of bits for an image is set to $128$, that is, $\mathrm{CBR} =\frac{128}{8\times 28 \times 28}\approx 0.02$, and the case of $\mathrm{CBR} =\frac{1024}{8\times 28 \times 28}\approx 0.16$ is also simulated. For the CIFAR-$10$ dataset, we adopt two stages of Swin Transformer blocks with $[C_1, C_2, C]=[128,256,48]$-dimensional intermediate concatenated features and equivalently $\mathrm{CBR}=\frac{5000}{8\times 3\times 32 \times 32} \approx 0.2 $. 
Besides, the learning rate and batch size are default as $10^{-3}$ and $64$ for these two datasets, respectively. 

\begin{table}[htbp]
  \centering
  \caption{Structures of encoder and decoders for \texttt{SemanticBC-TriRL} on MNIST dataset}
  \renewcommand{\arraystretch}{1.4}
  \begin{tabular}{m{1.7cm} m{1.8cm}  m{1.8cm}  m{1.8cm} }  	 
    \hline
    \textbf{Modules} & \textbf{Layer} & \textbf{Output imensions} & \textbf{Activation}\\ 
    \hline
    Input & Image $\boldsymbol{m}$ & $(T,1,28,28)$ & $\backslash$ \\  
    \hline
    \multirow{4}{1.7cm}{Encoder}& Conv layers & $(T,64,7,7)$ & ReLU\\ 
     & Dense layers & $(T,128)$ & $\backslash$\\ 
    & Sample policy  &  $(T,128)$ & Gaussian policy \\ 
    & $Q_{\rm{TX}}(\cdot)$& $(T,128)$ & Binary\\
    \hline
    Channel & AWGN/Rayleigh/Rician & & \\ 
    \hline
    \multirow{3}{1.7cm}{Reconstruction Decoder} & $Q_{\rm{RX}_n}(\cdot)$ & $(T,64)$ & $\backslash$\\
    & Dense layers & $(T,64,7,7)$ & $\backslash$\\
    & conv layers & $(T,1,28,28)$ & ReLU, Sigmoid\\
    \hline
    \multirow{3}{1.7cm}{Classification Decoder} & $Q_{\rm{RX}_n}(\cdot)$ & $(T,64)$ & $\backslash$\\
    & Dense layer & $(T,32)$ & ReLU\\
    & Dense layer & $(T,10)$ & Softmax\\
    \hline
  \end{tabular}
  \label{tab: mnist-structure}
\end{table}

\begin{table}[htbp]
  \centering
  \caption{Structures of encoder and decoders for \texttt{SemanticBC-TriRL} on CIFAR-$10$ dataset}
  \renewcommand{\arraystretch}{1.4}
  \begin{tabular}{m{1.35cm} m{2.73cm}  m{1.5cm}  m{1.5cm}}  	 
    \hline
    \textbf{Modules} & \textbf{Layer} & \textbf{Dimensions} & \textbf{Activation}\\ 
    \hline
    Input & Image $\boldsymbol{m}$ & $(T,32,32,3)$ & $\backslash$ \\  
    \hline
        \multirow{5}{1.35cm}{Encoder}& Swin Transformer block & $(T,16,16,128)$ & $\backslash$\\ 
     & Swin Transformer block & $(T,8,8,256)$ & $\backslash$\\ 
     & Sample policy  &  $(T,128)$ & Gaussian policy \\ 
     & Dense layer & $(T,3072)$ &$\backslash$\\ 
    & $Q_{\rm{TX}}(\cdot)$& $(T,5000)$ & Binary\\ 
    \hline
    Channel & AWGN/Rayleigh/Rician & & \\ 
    \hline
        \multirow{3}{1.35cm}{Reconstruction Decoder} & $Q_{\rm{RX}_n}(\cdot)$ & $(T,3072)$ & $\backslash$\\
    & Swin Transformer block  & $(T,8,8,256)$ & $\backslash$\\
    & Swin Transformer block  & $(T,16,16,128)$ & ReLU, Sigmoid\\
    \hline
    \multirow{3}{1.35cm}{Classification Decoder} & $Q_{\rm{RX}_n}(\cdot)$ & $(T,3072)$ & $\backslash$\\
    & Dense layers & $(T,128)$ & ReLU\\
    & Dense layer & $(T,10)$ & Softmax\\
    \hline
  \end{tabular}
  \label{tab: cifar10-structure}
\end{table}

\textbf{Comparsion Schemes:} We compare the performance between \texttt{SemanticBC-TriRL} and the following schemes including
\begin{itemize}
    \item \textbf{JSCC}: JSCC share the same encoder and reconstruction decoder as deep JSCC in \cite{bourtsoulatze2019deep}. Meanwhile, we enhance the convolutional neural network (CNN)-based JSCC scheme to support multi-task semantic BC by adding an extra classification decoder, thus enabling both image reconstruction and classification tasks. 
    \item \textbf{Deep JSCC}: Deep JSCC \cite{lyu2024semantic} is an end-to-end semantic Communication model for image recovery and classification, and we extend it to one-to-many semantic BC scenario for comparision. 
    \item \textbf{SemanticBC-TriRL-EW}: SemanticBC-TriRL-EW has the same DNN structure as \texttt{SemanticBC-TriRL} but the loss weight of it for each task is always adopt equal weighting (EW) during training, \ie, $w_n=\frac{1}{N}$. 
    \item \textbf{BPG+LDPC}\cite{richardson2018design}: The BPG codec\footnote{https://github.com/def-/libbpg} offers efficient image compression, and the 5G LDPC codec\footnote{https://github.com/NVlabs/sionna/tree/main/sionna/fec/ldpc} provides efficient error correction by using a sparse matrix design for channel coding. Thus, we expand this communication mechanism to BC scenario, and follow the standard 5G LDPC codes with different coding rates and quadrature amplitude modulations (QAM). We use combination of $(3072,4608)$ and LDPC codes (which corresponding to $2/3$ rate codes) with $4$-QAM and $16$-QAM digital modulation schemes. 
     To accomplish the classification task, we design a CNN-based classifier for the MNIST dataset and utilize a pre-trained ResNet20 that is specifically fine-tuned for the CIFAR-10 dataset as the classifier. 
\end{itemize}
It is worth noting that,  for the sake of fair comparison, while the BPG+LDPC method employs a dedicated quantization approach, we also apply 1-bit quantization to other DNN-based comparison algorithms with default quantization resolutions as our method.

\subsection{Numerical Results and Analysis}
\label{sec: Numerical Results and Analysis}

\subsubsection{Comparative Analysis with Baselines}
\label{sec: Comparative Analysis With Baselines:}

\begin{figure*}[hbt]
    \centering
    \subfigure[AWGN]{
        \includegraphics[width=0.483\textwidth]{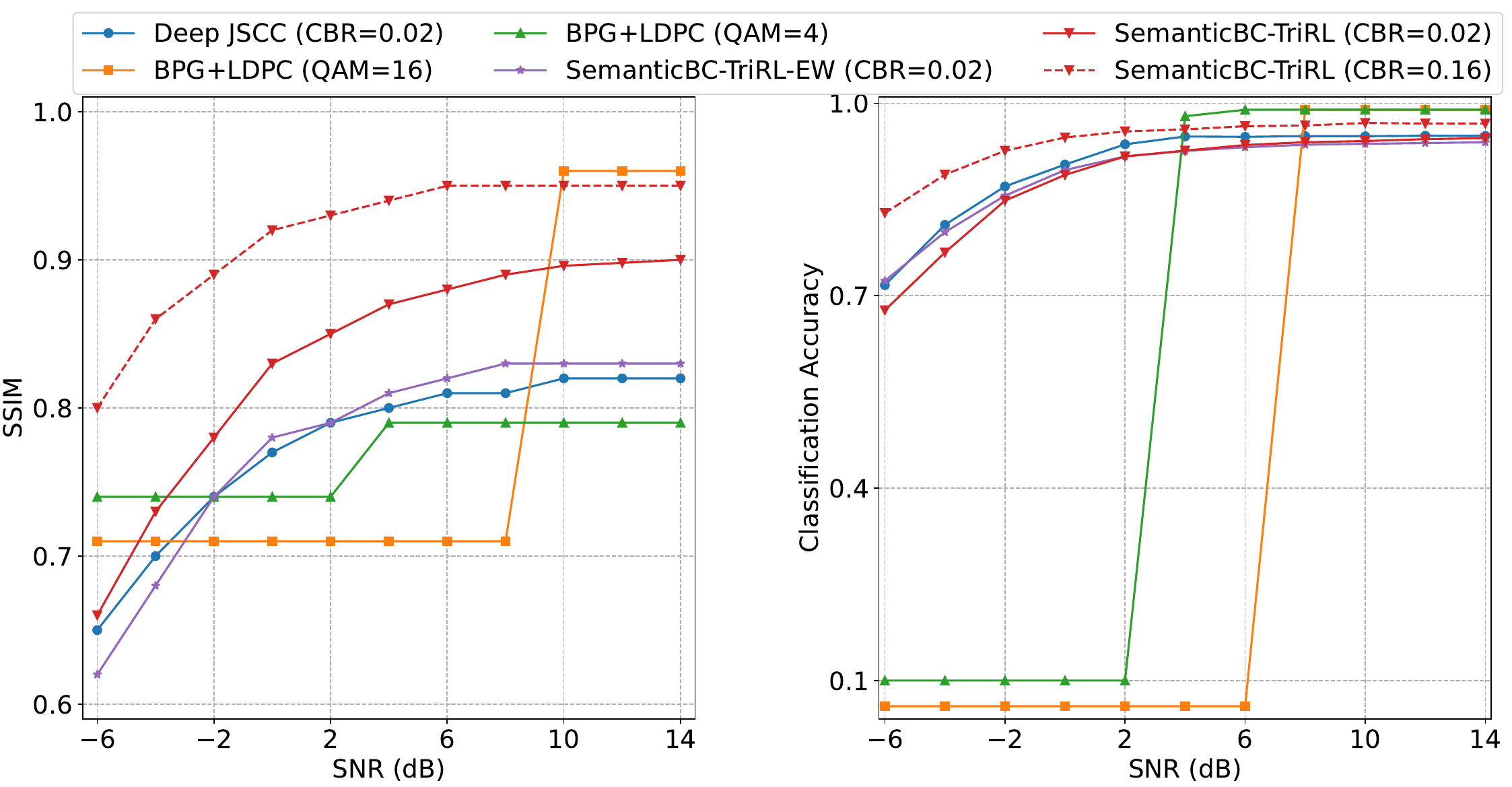}  
        \label{fig: mnist-awgn-ssim-acc}
    }
    \subfigure[Rayleigh fading]{
        \includegraphics[width=0.483\textwidth]{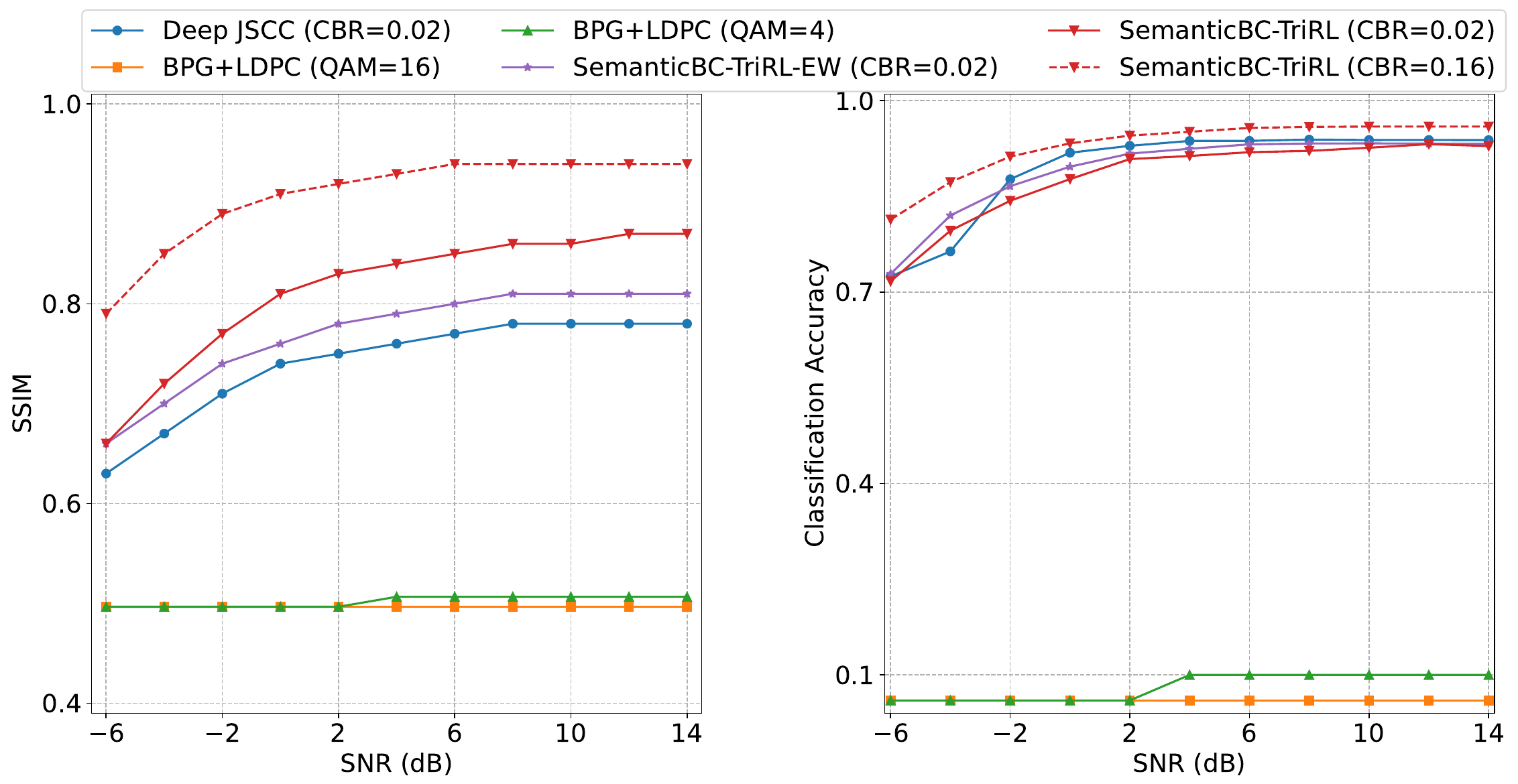}
        \label{fig: mnist-fading-ssim-acc}
    }    \caption{Comparison of semantic BC schemes in terms of SSIM and classification accuracy for MNIST image transmission under  AWGN and Rayleigh fading channels.}
    \label{fig: mnist-ssim-acc}
\end{figure*}

\begin{figure*}[hbt]
    \centering
    \subfigure[Original/(PSNR(dB) \& SSIM)]{
        \includegraphics[width=0.20\textwidth]{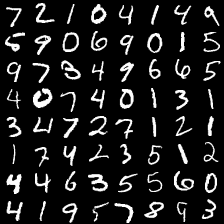} 
    }
    \subfigure[AWGN/($16.58$ dB \& $0.88$)]{
        \includegraphics[width=0.20\textwidth]{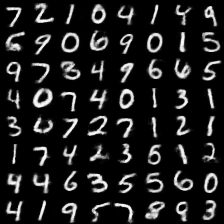}
    }
    \subfigure[Rayleigh/($15.67$ dB \& $0.85$)]{
        \includegraphics[width=0.20\textwidth]{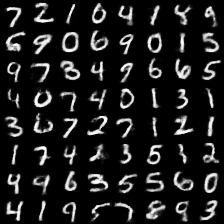} 
    }
    \subfigure[Rician/($15.80$ dB \& $0.86$)]{ 
        \includegraphics[width=0.20\textwidth]{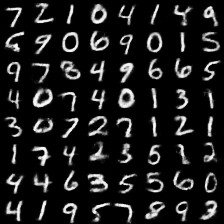}
    }
    \caption{Examples of visual comparison for MNIST image transmission through AWGN, Rayleigh, and Rician fading channels under $0$ dB with the proposed \texttt{SemanticBC-TriRL}.} 
    \label{fig: mnist-images}  
\end{figure*} 

\begin{figure*}[hbt]
    \centering
        \subfigure[BPG+LDPC (QAM = 4)/($9.05$ dB \& $0.71$))]{
        \includegraphics[width=0.20\textwidth]{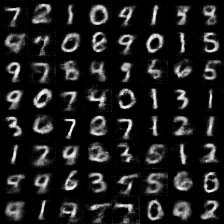} 
    }
    \subfigure[BPG+LDPC (QAM = 16)/($9.67$ dB \& $0.74$)]{
        \includegraphics[width=0.20\textwidth]{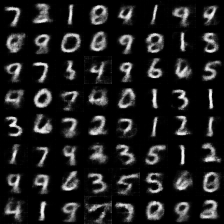} 
    }
    \subfigure[Deep JSCC/($14.46$ dB \& $0.77$)]{
        \includegraphics[width=0.20\textwidth]{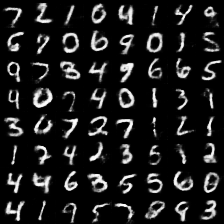} 
    }
    \subfigure[SemanticBC-TriRL-EW/($14.34$ dB \& $0.78$)]{
        \includegraphics[width=0.20\textwidth]{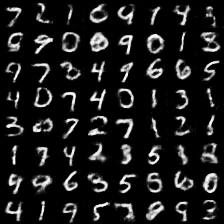}
    }    \caption{Examples of visual comparison for MNIST image transmission through AWGN channel under $0$ dB for comparsion schemes.}

    \label{fig: mnist-compare-images}
\end{figure*}

    We first evaluate the performance of the proposed \texttt{SemanticBC-TriRL} on the MNIST dataset, and compare it with several representative baselines, including ``Deep JSCC'', ``BPG+LDPC'' and ``SemanticBC-TriRL-EW'' under both AWGN and Rayleigh fading channels, as depicted in Fig. \ref{fig: mnist-ssim-acc}. From Fig. \ref{fig: mnist-awgn-ssim-acc}, it can be observed that under the AWGN channel,  \texttt{SemanticBC-TriRL} significantly outperforms other DNNs based methods such as Deep JSCC, SemanticBC-TriRL-EW, as well as the conventional coding scheme BPG+LDPC coding. In particular, for DNNs based semantic communication methods under a fixed quantization budget of $B=128$ bits (\ie, $\mathrm{CBR}= 0.02$), our method consistently achieves higher SSIM values and classification accuracies, especially in the low SNR regime ($\mathrm{SNR} < 8 $ dB), which is primarily attributed to the key semantic features can be more effectively encoded and transmitted to execute specific tasks. 
    
    Notably, the proposed \texttt{SemanticBC-TriRL} demonstrates a smooth and gradual performance degradation as SNR decreases, whereas the conventional ``BPG+LDPC'' scheme exhibits a ``cliff effect'' (characterized by a sudden drop in performance below a certain $\mathrm{SNR}$ threshold). This phenomenon is clearly manifested in both the SSIM and classification accuracy curves, where BPG+LDPC fails catastrophically under low SNR conditions. 
    Moreover, although BPG+LDPC achieves near-optimal performance in high SNR conditions, our method provides comparable performance in that regime while maintaining superior robustness at low SNRs. When the quantization bits are increased to $B=1024$ ($\mathrm{CBR} = 0.16$), both image quality and task accuracy are further improved, indicating that the richer semantic representation enabled by higher quantization resolution, thereby enhancing downstream task performance. 
    Moreover, the performance gap between \texttt{SemanticBC-TriRL} and its equal-weight counterpart, SemanticBC-TriRL-EW, highlights the effectiveness of the proposed task-aware weight assignment module in achieving a better balance between multi-task objectives. It is readily seen that the further incorporation of the weight assignment module leads to further improvement, and such an advantage arises from its ability to leverage the gradient norms of each decoder to obtain proper weights and find an optimal learning direction for the encoder. 

    In Fig. \ref{fig: mnist-fading-ssim-acc}, under the Rayleigh fading channel, the proposed \texttt{SemanticBC-TriRL} demonstrates well balanced performance across both tasks. Compared to the traditional BPG+LDPC scheme, which fails to adapt to fading channel and exhibits negligible performance variation across SNRs, \texttt{SemanticBC-TriRL} shows strong robustness and stable performance gains. Although Deep JSCC and SemanticBC-TriRL-EW show moderate resilience, their overall performance remains less balanced across tasks. Notably, at the same quantization level ($B=128$) under low $\mathrm{SNR}$ conditions ($0 \sim 2$ dB), \texttt{SemanticBC-TriRL} achieves up to $6.7\%$ improvement in SSIM, while SemanticBC-TriRL-EW provides only a slight advantage in classification accuracy (up to $0.95\%$). 
    These results further validate the effectiveness of our adaptive weighting strategy in achieving balanced and robust multi-task performance under varying channel conditions. 
 
\begin{figure*}[t!]
   \centering
   \setlength{\abovecaptionskip}{0.3cm} 
		\includegraphics[width=0.85\linewidth]{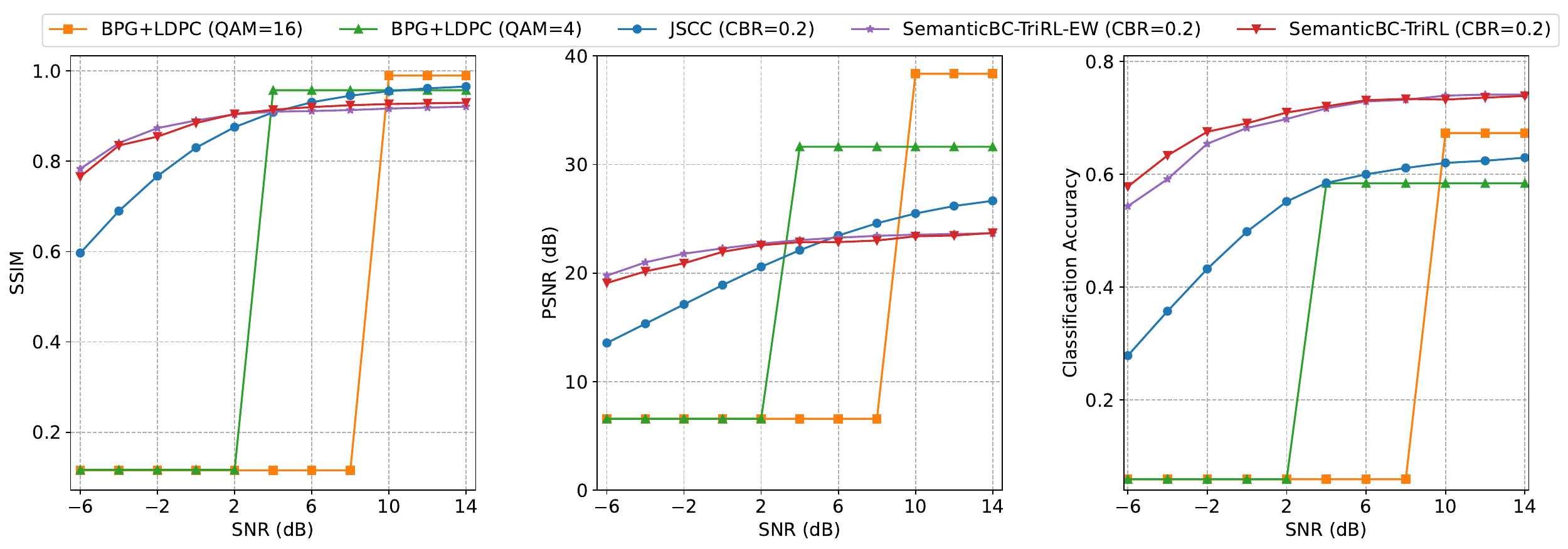}\\   
		\caption{Comparison of semantic BC schemes in terms of SSIM, PSNR and classification accuracy for CIFAR-10 image transmission under AWGN channel.} 
		\label{fig: cifar10-awgn-psnr-ssim-acc}      
\end{figure*} 
\FloatBarrier
\begin{figure*}[hbt]
    \centering
        \subfigure[Original/(PSNR(dB) \& SSIM)]{
        \includegraphics[width=0.20\textwidth]{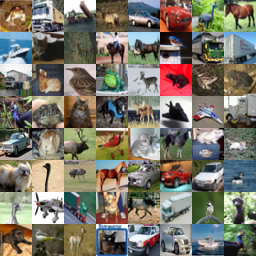} 
    }
        \subfigure[BPG+LDPC (QAM = 16)]{
        \includegraphics[width=0.20\textwidth]{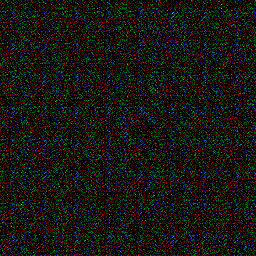} 
        \label{fig: cifar10-bpgldpc16qam}
    }
    \subfigure[BPG+LDPC (QAM = 4)/($31.76$ dB \& $0.9590$)]{
        \includegraphics[width=0.20\textwidth]{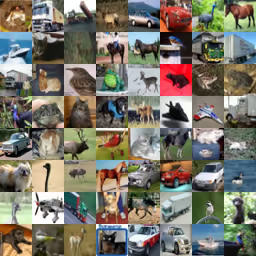} 
        \label{fig: cifar10-bpgldpc4qam}
    }
    \subfigure[JSCC/($22.74$ dB \& $0.9219$)]{
        \includegraphics[width=0.20\textwidth]{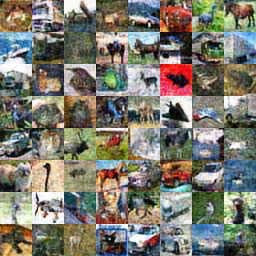} 
        \label{fig: cifar10-jscc}
    }
    \subfigure[SemanticBC-TriRL-EW/($23.18$ dB \&$0.9212$)]{
        \includegraphics[width=0.20\textwidth]{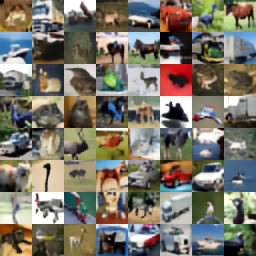}
        \label{fig: cifar10-ew}
    }    
    \subfigure[SemanticBC-TriRL/($23.11$ dB \&$0.9279$)]{
        \includegraphics[width=0.20\textwidth]{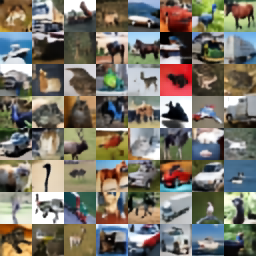}
        \label{fig: cifar10-me}
    } 
    \caption{Examples of visual comparison for CIFAR-10 image transmission through AWGN channel under $5$ dB for comparsion schemes.} 

    \label{fig: cifar10-images}
\end{figure*}

    Next, we also provide the visual results of image recovery on the MNIST dataset in Fig. \ref{fig: mnist-images} and Fig. \ref{fig: mnist-compare-images} under AWGN, Rayleigh, and Rician fading channels at $0$ dB. From Fig. \ref{fig: mnist-images}, it shows that \texttt{SemanticBC-TriRL} maintains high reconstruction quality across different channels at $0$ dB, with minimal visual degradation and consistently high SSIM and PSNR values, demonstrating strong robustness and adaptability to channel impairments. In Fig. \ref{fig: mnist-compare-images}, under AWGN at $0$ dB, all comparison schemes exhibit inferior performance in both visual quality and quantitative metrics compared to the proposed \texttt{SemanticBC-TriRL}. The performance degradation is particularly noticeable for traditional schemes, underscoring their vulnerability under low $\mathrm{SNR}$ conditions. 

    Moreover, we conduct additional evaluations on the more challenging CIFAR-$10$ dataset to investigate the generalization capability of the proposed \texttt{SemanticBC-TriRL}. As shown in Fig. \ref{fig: cifar10-awgn-psnr-ssim-acc}, under the same transmission bit budget, \texttt{SemanticBC-TriRL} and SemanticBC-TriRL-EW achieve comparable performance in terms of SSIM and PSNR for the image reconstruction task. But, \texttt{SemanticBC-TriRL} significantly outperforms other DNN-based baselines in classification accuracy, particularly in the low $\mathrm{SNR}$ regime (\ie, $\mathrm{SNR} < 4$ dB), demonstrating the effectiveness of the proposed PPO-based encoder optimization, which enables more robust semantic representation compared to traditional end-to-end JSCC.
    Although \texttt{SemanticBC-TriRL} exhibits slightly lower reconstruction quality than its equal-weight counterpart in certain cases, it yields notably higher classification accuracy, underscoring the importance of the proposed task-adaptive weight assignment module at the encoder side for balancing multiple downstream objectives. In addition, similar to the observations on the MNIST dataset, the traditional BPG+LDPC scheme continues to suffer from the ``cliff effect'' under AWGN channel, exhibiting abrupt performance degradation at low SNRs and limited robustness to channel impairments. 

    Furthermore, Fig. \ref{fig: cifar10-images} demonstrates the visual reconstruction on the CIFAR-$10$ dataset under $5$ dB AWGN channel, along with the corresponding PSNR and SSIM values. Among DNNs-based methods, the proposed \texttt{SemanticBC-TriRL} (Fig. \ref{fig: cifar10-me}) demonstrates competitive perceptual quality and preserves semantic content more faithfully. In contrast, the traditional BPG+LDPC exhibits pronounced performance disparities across different modulation settings. Specifically, the $16$ QAM (Fig. \ref{fig: cifar10-bpgldpc16qam}) fails to decode image in the presence of $5$dB noise, while $4$ QAM (Fig. \ref{fig: cifar10-bpgldpc4qam}) achieves higher PSNR and SSIM ($31.76$ dB and $0.9590$, respectively), highlighting the strong sensitivity to channel quality and modulation settings.

\subsubsection{Convergence and Scalability of \texttt{SemanticBC-TriRL}}
\label{sec: Convergence of SemanticBC-TriRL}

\begin{figure}[t]
   \centering
   \setlength{\abovecaptionskip}{0.3cm} 
		\includegraphics[width=0.80\linewidth]{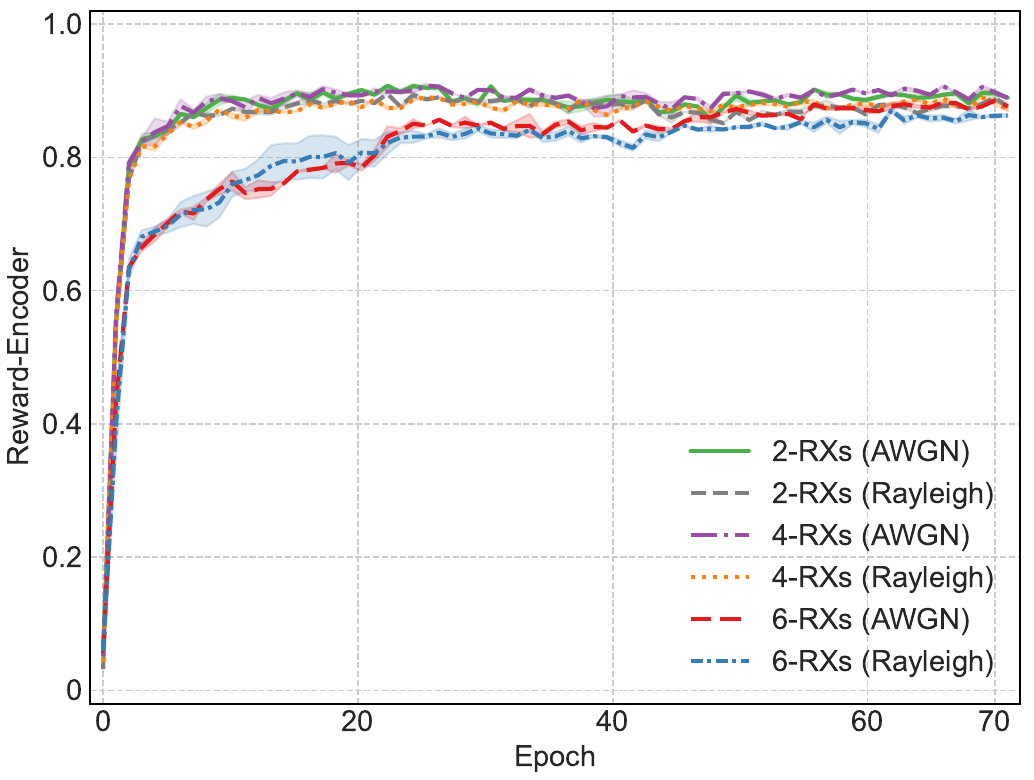}\\     
		\caption{Convergence of encoder of \texttt{SemanticBC-TriRL} for MNIST image transmission with two \rxs~ under AWGN, Rayleigh fading channels, wherein ${\rm RX}_1$ and ${\rm RX}_2$ execute image reconstruction and image classification tasks respectively.} 
		\label{fig: encoder-reward}  
\end{figure}

\begin{figure}[hbt]
    \centering
    \subfigure[Classification loss]{
        \includegraphics[width=0.40\textwidth]{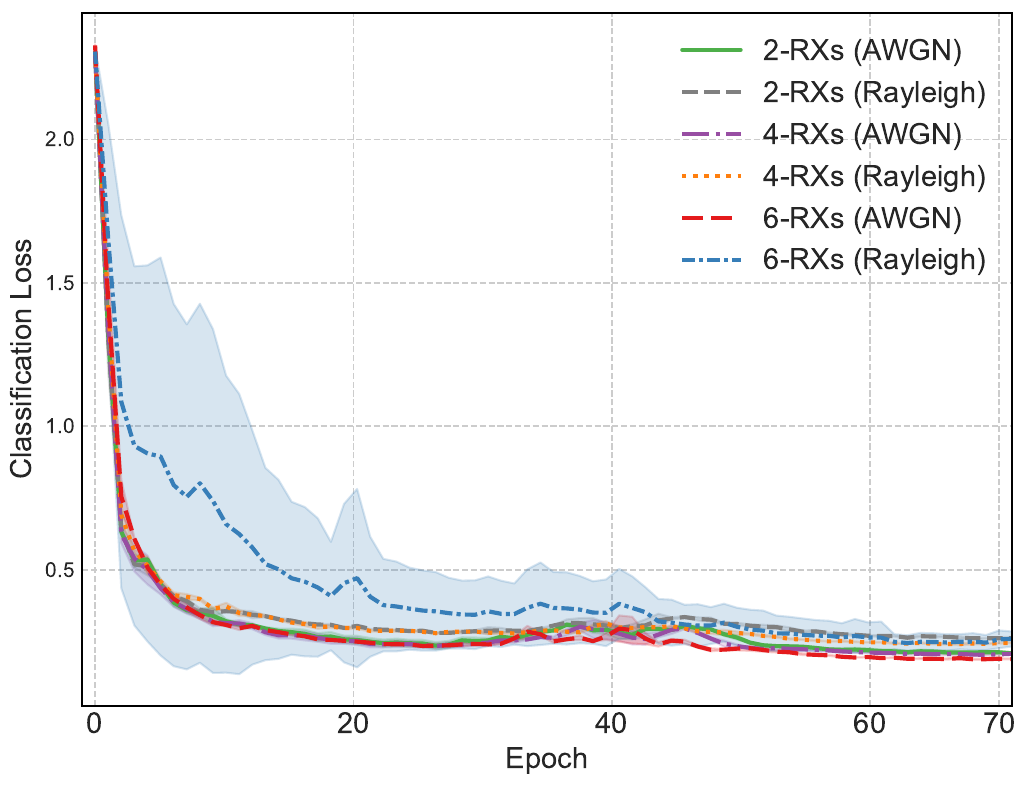}  
        \label{fig: mnist-cla-loss}
    }
    \subfigure[Reconstruction loss]{
        \includegraphics[width=0.40\textwidth]{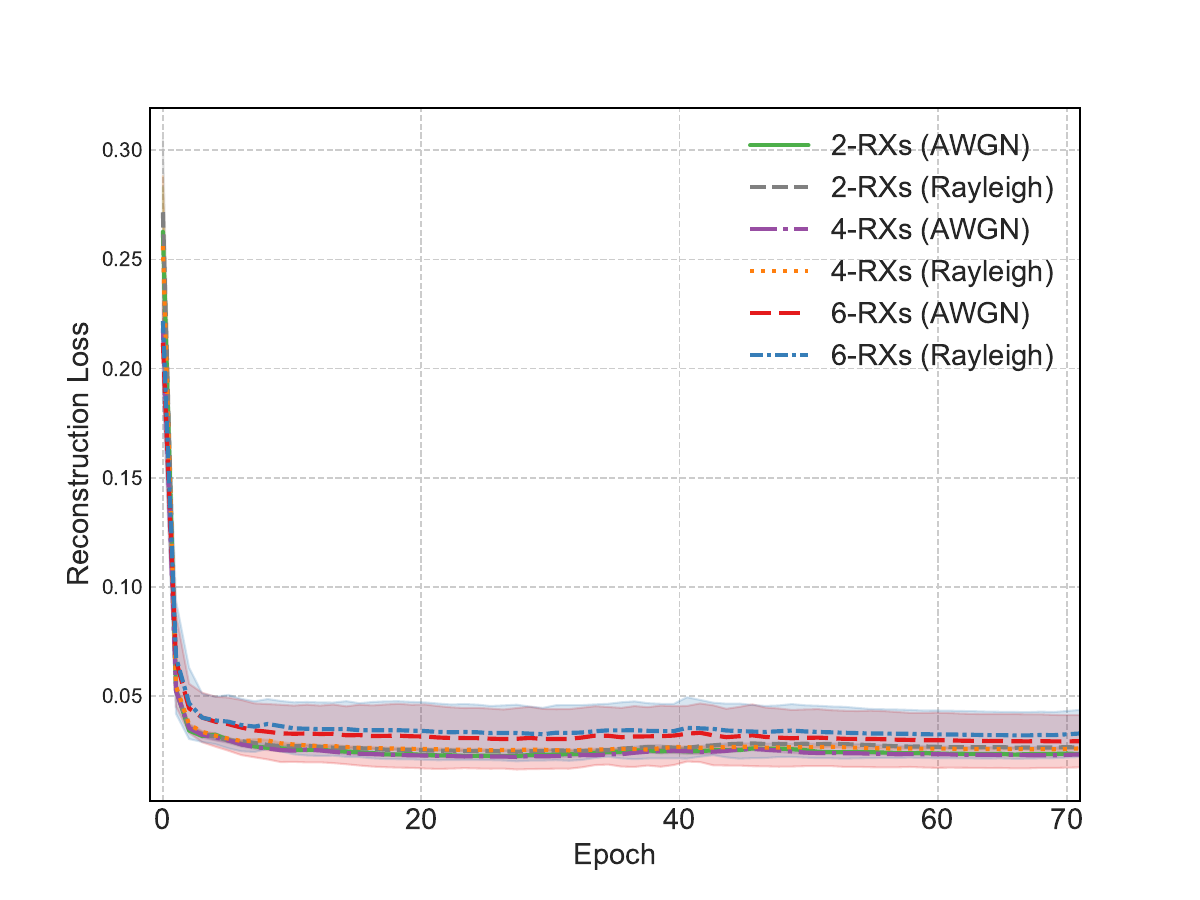}
        \label{fig: mnist-rec-loss}
    }    \caption{The training loss at decoder sides of \texttt{SemanticBC-TriRL} for MNIST image transmission through AWGN and Rayleigh fading channels.}
    \label{fig: mnist-cla-rec-loss}
\end{figure}

\begin{figure}[hbt]
   \centering
   \setlength{\abovecaptionskip}{0.3cm}  
		\includegraphics[width=0.80\linewidth]{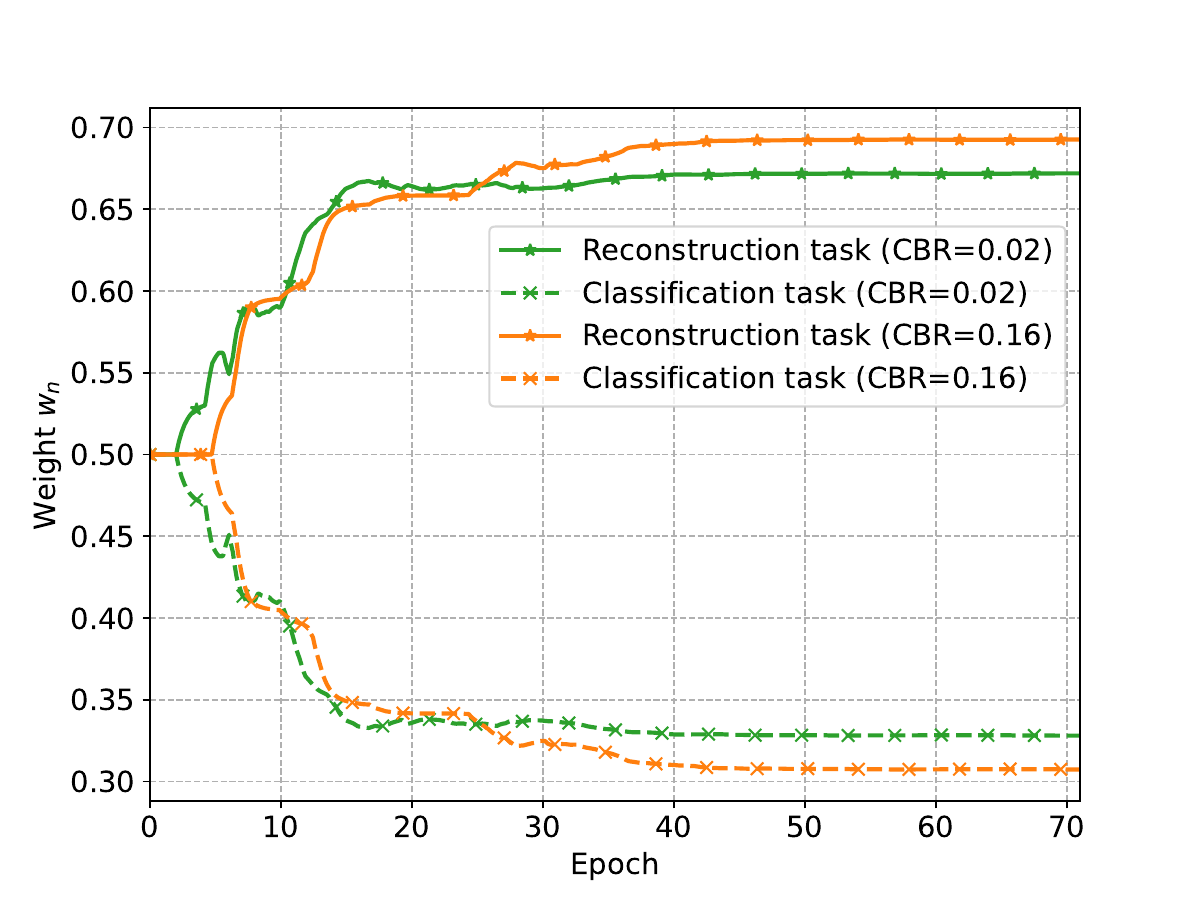}\\   
		\caption{Training weights evolution of our proposed \texttt{SemanticBC-TriRL} for MNIST image transmission in AWGN channel: reconstruction task vs. classification task.}  
		\label{fig: mnist-weights}      
\end{figure}

In this subsection, we investigate the convergence and scalability of the proposed \texttt{SemanticBC-TriRL}.
To further validate the effectiveness of the tri-level alternate learning framework, as illustrated in Fig. \ref{fig: encoder-reward} and Fig. \ref{fig: mnist-cla-rec-loss}, we present the training evolution for the encoder reward at \tx~under different channels with varying number of \rxs, as well as the classification and reconstruction losses of the respective decoders at \rxs. From Fig. \ref{fig: encoder-reward}, it can be observed that when the number of \rxs~is $2$ (\ie, $N=2$), the encoder reward defined in \eqref{eq: encoder reward} exhibits a sharp increase within the initial $10$ epochs under both channel conditions, followed by a gradual convergence, which demonstrates the stability and effectiveness of the PPO-based encoder optimization. Correspondingly, as shown in Fig. \ref{fig: mnist-cla-loss} and Fig. \ref{fig: mnist-rec-loss}, both the classification loss in \eqref{eq: cla-loss} and the reconstruction loss in \eqref{eq: rec-loss} decrease consistently and converge to low values, thereby validating the success of multi-task learning and the feasibility of the proposed tri-level alternating optimization strategy. 

On the other hand, as the number of \rxs~increases from $2$ to $6$ in Fig. \ref{fig: encoder-reward} and Fig. \ref{fig: mnist-cla-rec-loss}, a clear trend can be observed that both the encoder reward and the decoders' losses remain stable across different settings and ultimately converge to comparable performance regardless of the number of \rxs. 
This observation underscores the scalability of the proposed tri-level optimization framework among the \tx~side and heterogeneous \rxs. It is worth noting that the final convergence is similar across AWGN and Rayleigh fading channels, further confirming the robustness of the proposed learning framework under varying channel impairments. 

In addition, we testify to the evolution of task weights assigned during encoder training, as shown in Fig. \ref{fig: mnist-weights}. In both cases ($\mathrm{CBR}=0.02$ and $\mathrm{CBR}=0.16$), the weight assignment module consistently allocates a higher weight to the reconstruction task, indicating its larger influence on encoder optimization. Notably, this preference becomes more pronounced during the early stages of training under lower bandwidth conditions such as $\mathrm{CBR}=0.02$, as the reconstruction task demands finer-grained semantic representations, prompting the encoder to favor structural fidelity.

\subsubsection{Ablation Study}
\label{sec: Ablation Study}

\begin{figure}[t]
    \centering
    \subfigure[SSIM]{
        \includegraphics[width=0.40\textwidth]{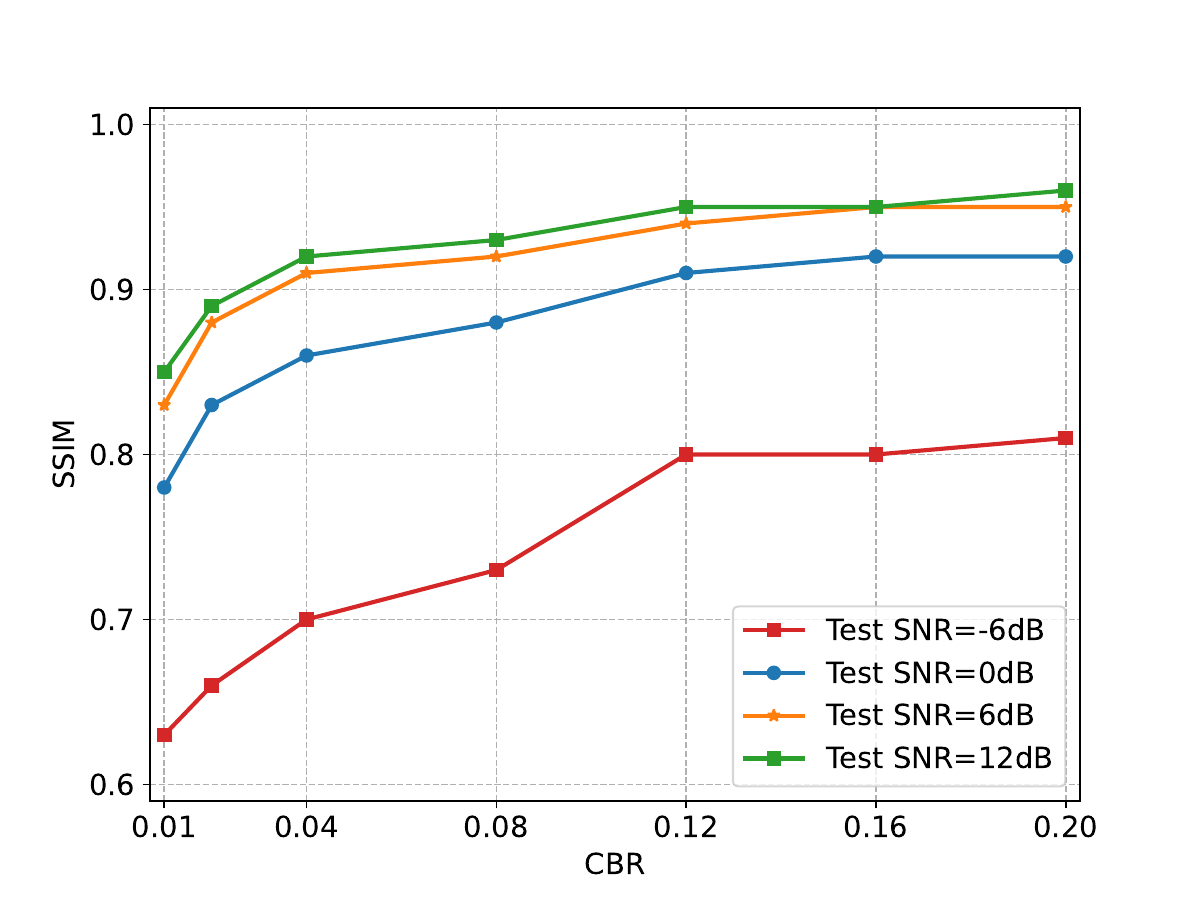}  
    }
    \subfigure[Classification accuracy]{
        \includegraphics[width=0.40\textwidth]{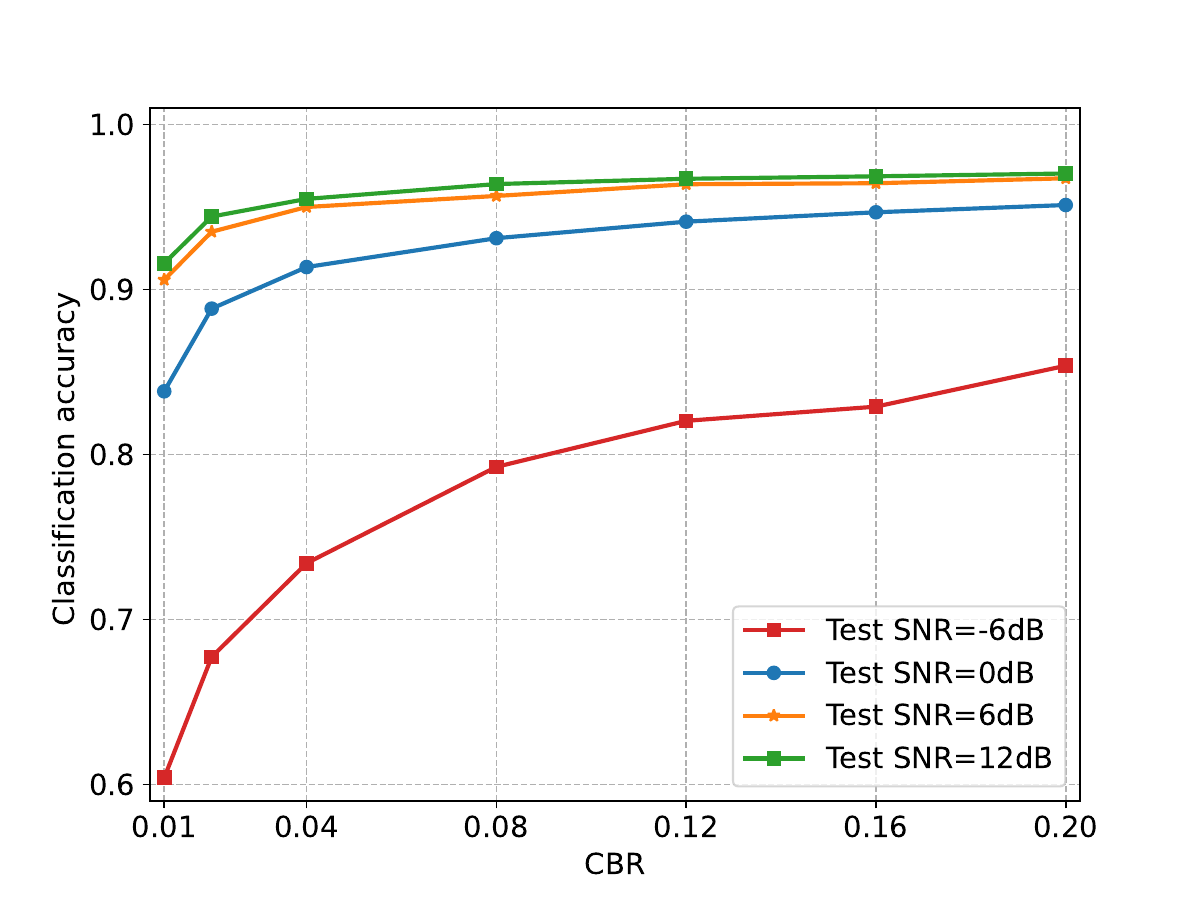}
    }    \caption{Impact of $\mathrm{CBR}$ on task performance of \texttt{SemanticBC-TriRL} for MNIST image transmission over AWGN channel}
    \label{fig: mnist-cbr}
\end{figure} 

Moreover, we further conduct an ablation study under varying $\mathrm{CBR}$ levels to evaluate the impact of bit budget on the two tasks at the \rxs~sides. As shown in Fig. \ref{fig: mnist-cbr}, both SSIM and classification accuracy improve steadily with increasing $\mathrm{CBR}$ under different SNRs, demonstrating that a higher bandwidth allows the encoder to transmit richer semantic information beneficial to both image reconstruction and classification.  Notably, the performance gains are more significant when $\mathrm{CBR}$ increases from very low values (e.g., $0.01$ to $0.05$), especially under low SNRs conditions (e.g., $-6$ dB), where additional bits substantially alleviate information bottlenecks.  When $\mathrm{CBR}$ exceeds approximately $0.12$, the performance tends to saturate, especially for higher SNRs, suggesting that further increases in bandwidth yield diminishing returns once the semantic representation becomes sufficiently expressive.

\section{Conclusions}
\label{sec: Conclusions}
In this paper, we have proposed a tri-level semantic BC framework, termed \texttt{SemanticBC-TriRL}, tailored for multi-task scenarios involving heterogeneous downstream tasks such as image classification and content reconstruction. This tri-level optimization framework achieves independent and alternating updates of the encoder and decoders, initiating updates from the decoders to the weight assignment module in a bottom-up manner, formulated as a multi-objective optimization problem with corresponding constraints. 
Specifically, at the first level, task-specific decoders are updated via supervised learning. Then, at the second level, the encoder is optimized using PPO to adapt to heterogeneous task requirements. 
Furthermore, to mitigate potential conflicts across multiple tasks, we have introduced a gradient aggregation-based weight assignment module at the third level to appropriately balance task weights for encoder optimization. Besides, we have also delved into the convergence analysis of \texttt{SemanticBC-TriRL} and provided non-asymptotic convergence. Extensive simulation results have demonstrated that the proposed semantic BC system achieves superior task performance and effectively maintains a stable trade-off. Moreover, we have also investigated the performance under varying numbers of \rxs~and different lengths of transmitted symbols, demonstrating superior scalability and adaptability with respect to the number of \rxs~and channel environments.
In the future, we will extend this optimization framework to accommodate a broader range of task types and integrate it with traditional multi-task learning methods, such as client selection and asynchronous transmission, to verify its generalization.

\bibliographystyle{IEEEtran}
\bibliography{SemanticBC-Multi-Task}

\begin{IEEEbiography}[{\includegraphics[width=1in,height=1.25in,clip,keepaspectratio]{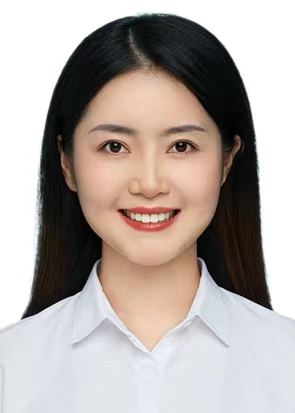}}]{Zhilin Lu} received the M.S. degree from North China Electric Power University, Beijing,
China, in 2021. She is currently pursuing the Ph.D. degree with the College of Information Science and Electronic Engineering, Zhejiang University, Hangzhou, China. Her research interests include semantic communication and deep reinforcement learning. 
\end{IEEEbiography}

\begin{IEEEbiography}[{\includegraphics[width=1in,height=1.25in,clip,keepaspectratio]{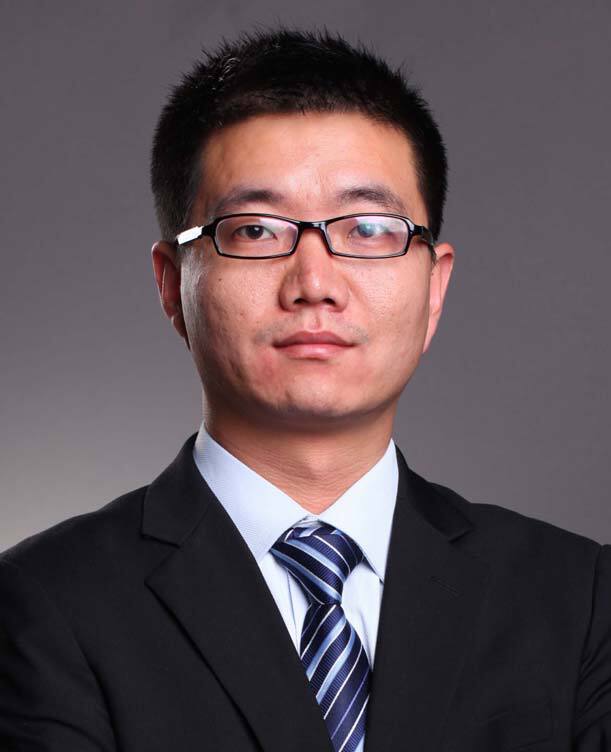}}]{Rongpeng Li}
    (Senior Member, IEEE) is currently an Associate Professor with the College of Information Science and Electronic Engineering, Zhejiang University, Hangzhou, China. He was a Research Engineer with the Wireless Communication Laboratory, Huawei Technologies Company, Ltd., Shanghai,
    China, from August 2015 to September 2016. He was a Visiting Scholar with the Department of Computer Science and Technology, University of Cambridge, Cambridge, U.K., from February 2020 to August 2020. His research interest currently focuses on networked intelligence for communications evolving (NICE). He received the Wu Wenjun Artificial Intelligence Excellent Youth Award in 2021. He serves as an Editor for \emph{China Communications}.
\end{IEEEbiography}	

\begin{IEEEbiography}[{\includegraphics[width=1in,height=1.25in,clip,keepaspectratio]{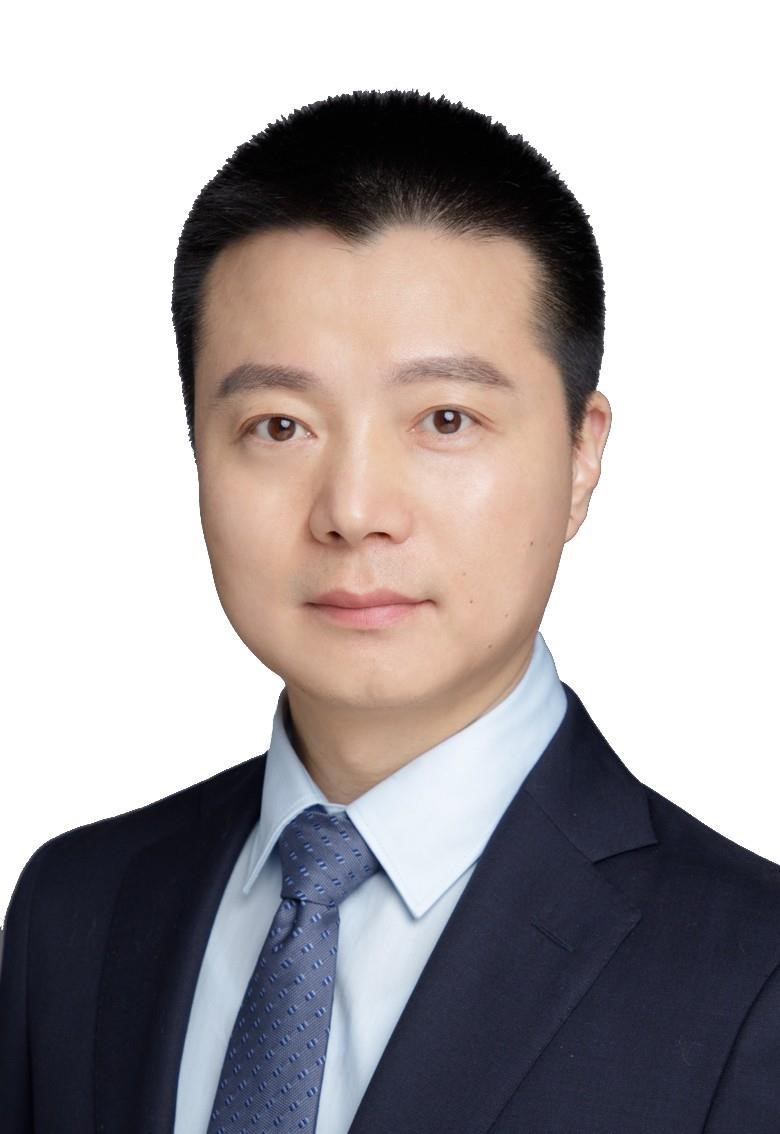}}]{Zhifeng Zhao}
    (Senior Member, IEEE) received the B.E. degree in computer science, the M.E. degree in communication and information systems, and the Ph.D. degree in communication and information systems from the PLA University of Science and Technology, Nanjing, China, in 1996, 1999, and 2002, respectively. From 2002 to 2004, he acted as a Post-Doctoral Researcher with Zhejiang University, Hangzhou, China, where his researches were focused on multimedia next-generation networks (NGNs) and softswitch technology for energy efficiency. From 2005 to 2006, he acted as a Senior Researcher with the PLA University of Science and Technology, where
    he performed research and development on advanced energy-efficient wireless router, \emph{ad-hoc} network simulator, and cognitive mesh networking test-bed.
    From 2006 to 2019, he was an Associate Professor with the College of Information Science and Electronic Engineering, Zhejiang University. Currently, he is with the Zhejiang Lab, Hangzhou as the Chief Engineering Officer. His research areas include software defined networks (SDNs), wireless network in 6G, computing networks, and collective intelligence. He was the Symposium Co-Chair of ChinaCom 2009 and 2010. He is the Technical Program Committee (TPC) Co-Chair of the 10th IEEE International Symposium on Communication and Information Technology (ISCIT 2010).
\end{IEEEbiography} 

\begin{IEEEbiography}[{\includegraphics[width=1in,height=1.25in,clip,keepaspectratio]{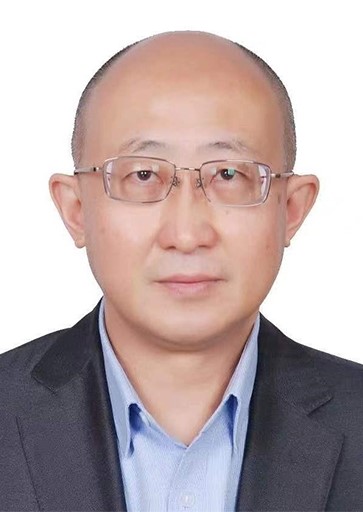}}]{Honggang Zhang}
   (Fellow, IEEE) is a Professor with the Faculty of Data Science, City University of Macau, Macau, China. He was the founding Chief Managing Editor of Intelligent Computing, a Science Partner Journal, as well as a Professor with the College of Information Science and Electronic Engineering, Zhejiang University, Hangzhou, China. He was an Honorary Visiting Professor with the University of York, York, U.K., and an International Chair Professor of Excellence with the Université Européenne de Bretagne and Supélec, France. He has coauthored and edited two books: Cognitive Communications: Distributed Artificial Intelligence (DAI), Regulatory Policy \& Economics, Implementation (John Wiley \& Sons) and Green Communications: Theoretical Fundamentals, Algorithms and Applications (CRC Press), respectively. His research interests include cognitive radio networks, semantic communications, green communications, machine learning, artificial intelligence, intelligent computing, and Internet of Intelligence. He is a co-recipient of the 2021 IEEE Communications Society Outstanding Paper Award and the 2021 IEEE Internet of Things Journal Best Paper Award. He was the leading Guest Editor for the Special Issues on Green Communications of the IEEE Communications Magazine. He served as a Series Editor for the IEEE Communications Magazine (Green Communications and Computing Networks Series) from 2015 to 2018 and the Chair of the Technical Committee on Cognitive Networks of the IEEE Communications Society from 2011 to 2012. He is the Associate Editor-in-Chief of China Communications.
\end{IEEEbiography}

\end{document}